\documentclass[10pt, letterpaper, amsmath,amssymb,epsfig,psfrag, twocolumn,geometry]{IEEEtran}
 \pdfoutput=1
\title{ The Streaming-DMT of Fading Channels}
\usepackage[tbtags]{amsmath} % defines many math commands and
\usepackage{amssymb}
\usepackage{verbatim} % get comment environment, + new verbatim
\usepackage{cite}
\usepackage{graphicx}
\usepackage{epsfig}
\usepackage{psfrag}
\usepackage{color}
\include{spacing}
\usepackage{xcolor}

\newcounter{actr}
{\begin{list}{(\alph{actr})}{\usecounter{actr}}}{\end{list}}

\newcounter{ictr}
{\begin{list}{(\roman{ictr})}{\usecounter{ictr}}}{\end{list}}

\newtheorem{remark}{Remark}
\newtheorem{thm}{Theorem}

\newtheorem{corol}{Corollary}
\newtheorem{prop}{Proposition}
\newtheorem{defn}{Definition}

{\noindent{\em Proof: }  \small \noindent}%
{\noindent\qed }
\newenvironment{new-proof}[1]
{{\em Proof }:\\}%
{ \noindent\qed }
\newcommand{\qed}{\rule[0.1ex]{1.4ex}{1.6ex}}

\newcommand{\defeq}{\stackrel{\Delta}{=}}

\newcommand{\mrm}{\mathrm}

\newcommand{\cA}{{\mathcal{A}}}

\newcommand{\cB}{{\mathcal{B}}}

\newcommand{\cC}{{\mathcal{C}}}

\newcommand{\cE}{{\mathcal{E}}}

\newcommand{\cF}{{\mathcal{F}}}

\newcommand{\cG}{{\mathcal{G}}}

\newcommand{\bH}{{\mathbf{H}}}

\newcommand{\bI}{{\mathbf{I}}}
\newcommand{\cI}{{\mathcal{I}}}

\newcommand{\cJ}{{\mathcal{J}}}

\newcommand{\CN}{{\mathcal{CN}}}
\newcommand{\Nt}{{\tilde{N}}}
\newcommand{\cO}{{\mathcal{O}}}

\newcommand{\cT}{{\mathcal{T}}}

\newcommand{\bx}{{\mathbf{x}}}

\newcommand{\bX}{{\mathbf{X}}}

\newcommand{\al}{\alpha}

\newcommand{\del}{\delta}

\newcommand{\eps}{\varepsilon}

\DeclareMathAlphabet{\mathbsf}{OT1}{cmss}{bx}{n}% bold sans serif
\DeclareMathAlphabet{\mathssf}{OT1}{cmss}{m}{sl}% slanted sans serif

\DeclareSymbolFont{bsfletters}{OT1}{cmss}{bx}{n}
\DeclareSymbolFont{ssfletters}{OT1}{cmss}{m}{n}
\DeclareMathSymbol{\bsfGamma}{0}{bsfletters}{'000}
\DeclareMathSymbol{\ssfGamma}{0}{ssfletters}{'000}
\DeclareMathSymbol{\bsfDelta}{0}{bsfletters}{'001}
\DeclareMathSymbol{\ssfDelta}{0}{ssfletters}{'001}
\DeclareMathSymbol{\bsfTheta}{0}{bsfletters}{'002}
\DeclareMathSymbol{\ssfTheta}{0}{ssfletters}{'002}
\DeclareMathSymbol{\bsfLambda}{0}{bsfletters}{'003}
\DeclareMathSymbol{\ssfLambda}{0}{ssfletters}{'003}
\DeclareMathSymbol{\bsfXi}{0}{bsfletters}{'004}
\DeclareMathSymbol{\ssfXi}{0}{ssfletters}{'004}
\DeclareMathSymbol{\bsfPi}{0}{bsfletters}{'005}
\DeclareMathSymbol{\ssfPi}{0}{ssfletters}{'005}
\DeclareMathSymbol{\bsfSigma}{0}{bsfletters}{'006}
\DeclareMathSymbol{\ssfSigma}{0}{ssfletters}{'006}
\DeclareMathSymbol{\bsfUpsilon}{0}{bsfletters}{'007}
\DeclareMathSymbol{\ssfUpsilon}{0}{ssfletters}{'007}
\DeclareMathSymbol{\bsfPhi}{0}{bsfletters}{'010}
\DeclareMathSymbol{\ssfPhi}{0}{ssfletters}{'010}
\DeclareMathSymbol{\bsfPsi}{0}{bsfletters}{'011}
\DeclareMathSymbol{\ssfPsi}{0}{ssfletters}{'011}
\DeclareMathSymbol{\bsfOmega}{0}{bsfletters}{'012}
\DeclareMathSymbol{\ssfOmega}{0}{ssfletters}{'012}

\renewcommand{\defeq}{\triangleq}

\newcommand{\rvh}{{\mathssf{h}}}    % h

\newcommand{\rvbH}{{\mathbsf{H}}}

\newcommand{\rvbn}{{\mathbsf{n}}}
\newcommand{\rvbv}{{\mathbsf{v}}}

\newcommand{\rvw}{{\mathssf{w}}}    % w
\newcommand{\svw}{w}

\newcommand{\rvx}{{\mathssf{x}}}    % x, random variable

\newcommand{\rvbx}{{\mathbsf{x}}}

\newcommand{\rvy}{{\mathssf{y}}}    % y

\newcommand{\rvby}{{\mathbsf{y}}}

\newcommand{\rvbz}{{\mathbsf{z}}}

\newcommand{\cH}{{\mathcal H}}

\newcommand{\hvw}{\hat{\rvw}}
\newcommand{\bvw}{\bar{\rvw}}
\newcommand{\Pe}{e_{\max}}
\newcommand{\Tk}{{T_{k}}}
\newcommand{\rvbY}{\mathbsf{{Y}}}
\newcommand{\rvbX}{\mathbsf{{X}}}
\newcommand{\rvbZ}{\mathbsf{{Z}}}
\newcommand{\svbY}{\mathbf{{Y}}}
\newcommand{\svbX}{\mathbf{{X}}}
\newcommand{\svbZ}{\mathbf{{Z}}}
\newcommand{\mbC}{\mathbb{C}}

\newcommand{\Nr}{N_\mrm{r}}
\renewcommand{\Nt}{N_\mrm{t}}
\author{Ashish~Khisti~\IEEEmembership{Member,~IEEE,} and Stark C. 
  Draper~\IEEEmembership{Member,~IEEE} \thanks{Part of this work was
    presented at the International Symposium on Information Theory,
    St. Petersburg, Russia, 2011~\cite{KhistiD11}.  }
  \thanks{{A. Khisti's work was supported by an NSERC (Natural
      Sciences Engineering Research Council) Discovery Grant and an
      Ontario Early Researcher Award. S. Draper's work was supported
      by the National Science Foundation under CAREER grant CCF
      0844539 and by an NSERC Discovery Grant. } }}% <-this % stops a space}

\begin{document}
\maketitle
\thispagestyle{plain}
\pagestyle{plain}

\begin{abstract}
We consider the sequential transmission of a stream of messages over a
block-fading multi-input-multi-output (MIMO) channel.  A new message
arrives at the beginning of each coherence block, and the decoder is
required to output each message sequentially, after a delay of $T$
coherence blocks.  In the special case when $T=1$, the setup reduces
to the quasi-static fading channel. We establish the optimal
diversity-multiplexing tradeoff (DMT) in the high
signal-to-noise-ratio (SNR) regime, and show that it equals $T$ times
the DMT of the quasi-static channel.  The converse is based on
utilizing the delay constraint to amplify a local outage event
associated with a message, globally across all the coherence
blocks. This approach appears to be new. We propose two coding schemes
that achieve the optimal DMT. The first scheme involves interleaving
of messages, such that each message is transmitted across $T$
consecutive coherence blocks.  This scheme requires the knowledge of
the delay constraint at both the encoder and decoder. Our second
coding scheme involves a sequential tree code and is delay-universal
i.e., the knowledge of the decoding delay is not required by the
encoder. However, in this scheme we require the coherence block-length
to increase as $\log\mrm{({SNR})}$, in order to attain the optimal
DMT. Finally, we discuss the case when multiple messages arrive at
uniform intervals {\em within} each coherence period. Through a simple
example we exhibit the sub-optimality of interleaving, and propose
another scheme that achieves the optimal DMT.
\end{abstract}


\section{Introduction}
Multimedia applications require real-time encoding of a source stream
and a sequential reconstruction of each source-packet by its playback
deadline.  Both the fundamental limits and optimal communication
techniques for such {\em streaming systems} can be very different from
classical communication systems.  In recent years there has been a
growing interest in characterizing information theoretic limits for
delay-constrained communication over wireless channels. When the
transmitter has channel state information (CSI), a notion of
delay-limited capacity can be defined~\cite{hanly-delay}. For slow
fading channels, the delay-limited capacity is achieved using channel
inversion at the transmitter~\cite{caire-delay}.  In absence of
transmitter CSI, an outage capacity can be
defined~\cite{shamai:98,tseViswanath:05}.  Unfortunately the
characterization of the outage capacity is in general a challenging
problem, even in point-to-point settings~\cite{verdu-myth}. A somewhat
coarse metric for studying the outage capacity is the
diversity-multiplexing tradeoff (DMT), first introduced
in~\cite{ZhengTse:03}.  The authors propose {\em diversity order} and
{\em multiplexing gain} as two fundamental metrics for communication
over a wireless channel, and establish a tradeoff between these for
quasi-static, multi-input-multi-output (MIMO) fading channels, in the
high signal-to-noise-ratio (SNR) regime.  A significant body of
literature on DMT now exists, see e.g.,~\cite[Chapter
  9]{tseViswanath:05}. Of particular interest in this work is the case
of $T$ independent parallel MIMO fading channels where the optimal DMT
equals $T$ times the DMT of the quasi-static MIMO fading channel, with
a suitably normalized multiplexing gain~\cite{ZhengTse:03}. Practical
code constructions for parallel fading channels have been proposed
in~\cite{yang,lu,mroueh}.  Interestingly, when the parallel channels
are correlated, the DMT analysis is far more intricate and only
special cases are known~\cite{coronel, bolsckei}.

In the present paper we consider the problem of real-time streaming
over a MIMO block-fading wireless channel.  We assume that the
transmitter observes a sequence of independent messages.  One message
arrives per coherence block, right at the start of the block.  The
input into the channel can depend on all the past messages, but not on
any future messages.  The decoder is required to output each message
with a maximum delay of $T$ coherence blocks.  When $T=1$, each
message sees only one fading realization and the setup reduces to the
quasi-static fading channel model.  In general, each message
experiences $T$ independent fading blocks; however it must be
multiplexed together with messages arriving in other coherence blocks.
We declare a message to be in outage if it cannot be decoded by its deadline. We
establish the optimal DMT in this streaming setting and show that it
equals $T$ times the DMT of the quasi-static MIMO fading
channel. The optimal DMT can be achieved by a simple
interleaving of messages across coherence blocks and transmitting
each message over $T$ parallel, MIMO fading channels. 
We also propose an
alternative tree-code that attains the optimal DMT.  In this scheme the delay constraint only needs to be revealed to
the decoder, and not to the encoder, and thus it is suitable for
applications where a common source stream must be transmitted to
multiple receivers with different decoding delays.
However in order to
achieve the optimal DMT, the coherence block length for tree-codes must
increase with $\log\left(\mrm{SNR}\right)$ and thus this scheme appears to require
long coherence periods.

Tree codes for streaming communication over a discrete
memoryless channels (DMC) have been studied previously
in~\cite{Sahai-1a,Sahai-2,Sahai-5,DraperStreaming,sukhavasi2011linear}.  These works,
however, consider maximum likelihood and universal decoders. In
contrast, our analysis of the tree code is based on a very different
outage analysis paired with a decision directed decoder. We express
our error probability as a sum of two terms --- one term decreases
exponentially in $\log(\mathrm{SNR})$ while the other decreases
exponentially in the coherence block-length. By suitably balancing the
two exponents we establish that our proposed scheme attains the
optimal DMT.  Another recent work, reference~\cite{cocco}, studies a
related setup when the transmitter sequentially observes a stream of
messages, but assumes that all the messages have a common deadline. A
variety of coding techniques such as adaptive joint encoding,
memoryless transmission, time sharing and superposition transmission
are compared in different delay and SNR regimes. Another followup
work~\cite{cocco2} considers the case when all the messages are
available to the encoder, but have different playback deadline at the
receiver. In contrast our proposed setup requires that each incoming
message must be reconstructed after a fixed decoding deadline, which
is relevant in applications such as real-time voice and video
streaming.  Finally in yet another related
work~\cite{Kittipiyakul,Kittipiyakul2}, the authors study the
transmission of bursty and delay-sensitive data source over a
constant-rate MIMO fading channel and establish an optimal operating
point on the DMT that balances the channel outage and
{\emph{delay-violation}} probabilities. However the results are valid
only for asymptotically large decoding delays. Furthermore it appears
that the coding techniques considered in these works do not retransmit
the same information bits across multiple coherence blocks, a key idea
exploited in the present paper.

In the rest of the paper we describe the system model in Section~\ref{sec:model}, and the main result, that characterizes the streaming-DMT in Section~\ref{sec:result}.  We provide the proof of the converse in Section~\ref{sec:Converse}. The coding schemes based on interleaving and tree-codes are presented in Section~\ref{sec:Coding}. We discuss extension to the case of multiple messages in Section~\ref{sec:multiple} and provide conclusions in Section~\ref{sec.conclusion}.

Throughout the paper we will use the following notation. Upper case bold-font will be reserved for matrices (e.g.,~$\bH$) whereas lower case bold-font (e.g.,~$\bx$) will be used for vectors. Scalar symbols will be denoted using lower case non-bold fonts. We will use the sans serif font for random variables e.g., $\rvx$.  A sequence of symbols $x_i, x_{i+1},\ldots, x_j$ will be denoted using the notation $x_i^j$. Throughout the paper the symbol $\doteq$ will be reserved to denote equality in the exponential sense i.e., we express, $f(\rho) \doteq \rho^b$, if $\lim_{\rho\rightarrow\infty}\frac{\log f(\rho)}{\log \rho} = b$ holds. The symbols $\stackrel{\cdot}{\le}$ and $\stackrel{\cdot}{\ge}$ will be defined in a similar fashion.

\section{Model}
\label{sec:model}

We consider an independent identically
distributed (i.i.d.) block fading channel model with a coherence
period of $M$:
\begin{equation}
\rvbY_k = \rvbH_k \cdot \bX_k +
\rvbZ_k,\label{eq:chan_Model}
\end{equation} 
where $k = 0,1,\ldots$, denotes the index of the coherence block of
the fading channel.  
The matrix $\rvbH_k \in \mathbb{C}^{\Nr \times \Nt}$
denotes the channel transfer matrix in coherence period $k$.
We assume that the transmitter has $\Nt$ transmit antennas
and the receiver has $\Nr$ receive antennas. 
 $${\rvbX_k = \left[\rvbx_k(1)~|~\ldots ~|~\rvbx_k(M)\right]}\in
\mbC^{N_t \times M}$$ is a matrix whose $j$-th column, $\rvbx_k(j)$,
denotes the vector transmitted in time-slot $j$ in the coherence block
$k$ and similarly $\rvbY_k \in \mbC^{N_r \times M}$ is a matrix whose
$j$-th column, $\rvby_k(j)$ denotes the vectors received in time-slot
$j$ in block $k$. The additive noise matrix is $\rvbZ_k \in \mbC^{N_r
  \times M}$. Thus~\eqref{eq:chan_Model} can also be expressed as,
\begin{equation}
\rvby_k(j) = \rvbH_k \cdot \rvbx_k(j) + \rvbz_k(j), \qquad j=1,\ldots,
M.\label{eq:chan_Model_2}
\end{equation} 
We assume that all entries of $\rvbH_k$ are sampled independently from
the complex Gaussian distribution\footnote{While we only focus on the
  Rayleigh channel model, our results easily extend to other channel
  models.}  with zero-mean and unit-variance i.e., $\CN(0,1)$.  The
channel remains constant during each coherence block and is sampled
independently across blocks.  All entries of the additive noise matrix
$\rvbZ_k$ are also sampled i.i.d.\ $\CN(0,1)$.  Finally the realization
of the channel matrices $\rvbH_k$ is revealed to the decoder, but not
to the encoder.

We assume an average (short-term) power constraint
${E[\sum_{i=1}^M||\rvbx_k(i)||^2] \le M \rho}$.  Note that $\rho$
denotes the transmit power which will serve as our SNR parameter.  
A delay-constrained streaming code is defined
as follows:
\begin{defn}[Streaming Code]
\label{def:Rate}
A rate $R$ streaming code with delay $T$, $\cC(R,T)$, consists of
\begin{itemize}
\item[1.] A sequence of messages $\{\rvw_k\}_{k\ge 0}$ each
  distributed uniformly over the set ${\cI_M=\{1,2,\ldots, 2^{MR}\}}$.
\item[2.] A sequence of encoding functions ${\cF_k : \cI_M^{k+1}
  \rightarrow {\mathbb{C}}^{\Nt \times M}}$,
\begin{equation}
\rvbX_k = \cF_k(\rvw_0,\ldots, \rvw_k),\quad
k=0,1,\ldots\label{eq:enc}
\end{equation} 
that maps the input message sequence to the channel input matrix ${\rvbX_k} \in {\mathbb{C}}^{\Nt \times M}$.
\item[3.] A sequence of decoding functions ${\cG_k :
  {\mathbb{C}}^{M(k+T)} \rightarrow \cI_M}$ that outputs message
  estimate $\hat{\rvw}_k$ based on the first $k+T$ observations, i.e.,
\begin{equation}
\hat{\rvw}_k = \cG_k(\rvbY_0,\ldots, \rvbY_{k+T-1}),\quad
k=0,1,\ldots\label{eq:dec}
\end{equation}
\end{itemize}
%\hfill$\Box$
\end{defn}

Fig.~\ref{fig:model} illustrates such a setup for the case when
$T=2$. One message $\rvw_k$ arrives at the start of each coherence
block. The codeword transmitted in block $k$, $\rvbX_k(\rvw_0^k)$ can
depend on all the past messages, but not on any future messages. Since
$T=2$, the receiver must decode message $\rvw_k$ at the end of
coherence block $k+1$ i.e., $\hat{\rvw_k} = \cG_k(\rvbY_0,\ldots,
\rvbY_{k+1})$.

\begin{figure*}
\begin{center}
\includegraphics[scale=0.8]{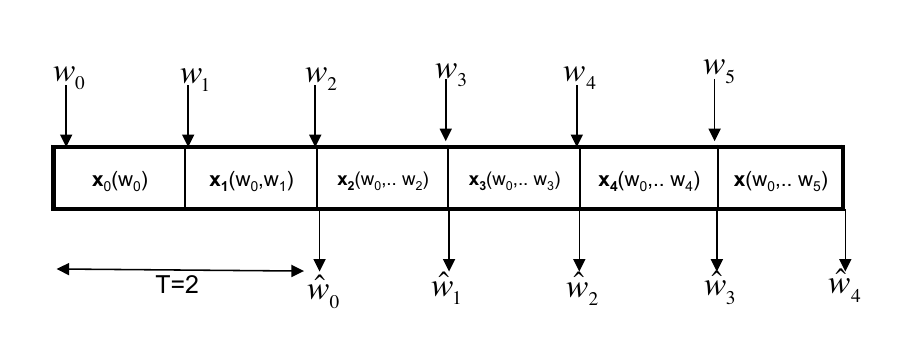}
\caption{Proposed Streaming Model. One new message arrives at the start of each coherence
block. The message stream is encoded sequentially  and each message
needs to be output at the receiver after $T$ coherence blocks. In the above figure $T=2$.  }
\label{fig:model}
\end{center}
\end{figure*}

{We now define the diversity-multiplexing tradeoff (DMT) associated
  with the streaming code $\cC(R,T)$.  The error probability for the
  $k$-th message is $\Pr[\rvw_k \neq \hvw_k]$ where $\hat{\rvw}_k$ is
  the decoder output~\eqref{eq:dec} and the error probability is
  averaged over the random channel gains.  In the DMT both error
  probability and rate are studied as a function of the SNR parameter
  $\rho$.  Let $\Pe(\rho) = \sup_{k\ge 0} \Pr[\hat{\rvw}_k\neq\rvw_k]$
  denote the worst-case error probability of a rate-$R(\rho)$ code. A
  DMT tradeoff~\cite{ZhengTse:03} of $(r,d)$ is said to be achievable
  with delay $T$ if there exists a sequence of codebooks
  $\cC(R(\rho),T)$ achieving $\Pe(\rho)$ such that
\begin{equation}
r = \lim_{\rho \rightarrow \infty} \frac{R(\rho)}{\log\rho}, \quad d =
\lim_{\rho\rightarrow\infty}\frac{-\log \Pe(\rho)}{\log \rho}.
\end{equation}
 Of interest, is the optimal diversity-multiplexing tradeoff, denoted
 by $d_T(r)$.} \footnote{We caution the reader that in the above
   discussion $\Pe(\rho)$ is {\em not} the maximum error probability
   with respect to a single realization of the fading state
   sequence. This later quantity is clearly $1$ as in any sufficiently
   long realization there will eventually be a block fade that induces
   an outage. In our definition we fix an index, $k$, and find the
   error probability $\Pr[\hat{\rvw}_k \neq \rvw_k]$ averaged over the
   channel gains. We subsequently search for the index $k$ with the
   maximum error probability. E.g., for time-invariant coding schemes
   operating in the steady-state regime, due to symmetry,
   $\Pr[\hat{\rvw}_k \neq \rvw_k]$ will not depend on $k$. }

One class of wireless systems that motivates this model is
frequency-hopping orthogonal frequency-division multiple access
(OFDMA) systems.  Here the frequency bands (or ``sub-channel'')
allocated to a user are changed at regular intervals and in a
randomized fashion.  In frequency-selective channels the randomized
sub-channel allocation means that the channels the the user
experiences pre- and post-hop are approximately independent, as long
as the expected recurrence time of a particular sub-channel exceeds
the coherence time of the channel.  In this setting we can
characterize the playback deadline in terms of the number of hops $T$
until the message $\rvw_k$ must be estimated by the receiver.  This
can be translated into a number of channel uses, $TM$, where $M$
denotes the number of symbols transmitted in each hop, and hence into
time.  In practical systems the value of $M$ may be fixed and thus not
under the control of the application.

\section{Main Result}
\label{sec:result}
The optimal tradeoff between diversity and multiplexing (DMT) for the
quasi-static fading channel was characterized
in~\cite{ZhengTse:03}. We reproduce the result below for the
convenience of the reader.
\begin{thm}{(Zheng and Tse, \!\cite{ZhengTse:03})}
For the quasi-static fading channel
\begin{align}
\rvby(t) = \rvbH\cdot\rvbx(t) + \rvbz(t)
\end{align}
where the entries of ${\rvbH \in {\mathbb C}^{\Nr \times \Nt}}$ are
sampled i.i.d.\ $\CN(0,1),$ the optimal DMT tradeoff $d_1(r)$ is a
piecewise linear function connecting the points $(k, d_1(k))$ for $k =
0,1,\ldots, \min(\Nr, \Nt)$ where
\begin{equation} 
d_1(k) = (\Nr-k)(\Nt-k). \label{eq:quasi-static_DMT}
\end{equation} 
\hfill$\Box$
\label{thm:DMT_Zheng}
\end{thm}

In our analysis the following  generalization of the quasi-static DMT
to $L$ parallel channels, see~\cite[Corollary 8]{ZhengTse:03}\cite{yang,lu} is useful.

\begin{corol}
Consider a collection of $L$ parallel 
quasi-static fading channel 
\begin{align}
\rvby_l(t) = \rvbH_l\cdot\rvbx_l(t) + \rvbz_l(t), \quad l=1,\ldots, L
\end{align}
where the entries of ${\rvbH_l\in {\mathbb C}^{\Nr \times \Nt}}$ are all sampled i.i.d. $\CN(0,1)$. The 
DMT tradeoff is given by $d^\mrm{\parallel}_L(r) = L\cdot d_1\left(\frac{r}{L}\right)$
for any $r \in (0, L\min(\Nr, \Nt))$.

\hfill$\Box$
\label{corol:parallel}
\end{corol}

\iffalse
We provide a proof of Corollary~\ref{corol:parallel} in
Appendix~\ref{app:parallel}.  The DMT for parallel fading channels was
also 
\fi

Our main result establishes the optimal DMT for a block fading channel
model with a delay constraint of $T$ coherence blocks.

\begin{thm}
The optimal DMT tradeoff for a streaming code
(cf.~Definition~\ref{def:Rate}) with a delay of $T$ coherence blocks
is given by $d_T(r) = T\cdot d_1(r)$, where $d_1(r)$ is the optimal
DMT of the underlying quasi-static fading
channel~\eqref{eq:quasi-static_DMT}.

\hfill$\Box$
\label{thm:streaming_DMT}
\end{thm}

Comparing the results of Theorem~\ref{thm:streaming_DMT} with that of
Corollary~\ref{corol:parallel}, we observe that the DMT of a streaming
source under a delay constraint of $T$ coherence blocks is identical
to the DMT of a system with $T$ independent and parallel MIMO channels
if the rate of the latter system is suitably normalized. Indeed, one
of our achievability schemes exploits this connection. We show that
the DMT can be achieved by interleaving messages in a suitable manner
to reduce the system to a parallel channel setup.  However, the
converse does not follow from earlier results since the length-$T$
playback deadlines of successive messages are only partially
overlapping.  We present a new approach that addresses the overlapping
character of the playback deadlines.  The technique is specific to the
streaming setup and appears novel.

\begin{remark}
In our system model, we assumed that the coherence block-length $M$
can be arbitrarily large. It is well known~\cite{ZhengTse:03} that for
quasi-static channel, as well as its extension to $L$ parallel
channel~\cite{lu}, the DMT in Theorem~\ref{thm:DMT_Zheng} holds for
any coherence block-length $M\ge \Nr + \Nt-1$. In a similar fashion
our result in Theorem~\ref{thm:streaming_DMT} holds for any $M \ge \Nr
+\Nt-1$. In particular the converse in Section~\ref{sec:Converse}
holds for any $M$. The interleaved coding scheme in
Section~\ref{sec:Coding-Interleaving} reduces the setup to parallel
channels and applies to any $M \ge \Nr +\Nt-1$. However this scheme
requires the knowledge of $T$ at both the encoder and decoder. Our
second coding scheme, which is based on a tree code and only requires
the knowledge of $T$ at the decoder, does require $M$ to be
sufficiently large. In particular our analysis for this scheme
requires that $M$ must increase as $\log\mrm{SNR}$ to achieve the
optimal DMT.
\end{remark}

\section{Converse}  
\label{sec:Converse}

In this section we establish a lower bound on the error probability
for any  streaming code in
Definition~\ref{def:Rate}.  We thereby upper bound the achievable DMT.
In particular we show that
$$\Pr[\mbox{error}] \stackrel{.}{\ge} \rho^{-Td_1(r)}$$ where $d_1(r)$
is the DMT tradeoff associated with a single-link MIMO channel.  For
the purpose of establishing a contradiction, we will assume that a DMT
{\em better} than $d_T(r)$ is achievable, say $T \cdot d_1(r -
\delta)$.  We show that for any $\delta > 0$ a contradiction will
build up if we operate the system over a sufficiently large number of
blocks $N$.  The smaller $\delta$ is the longer it takes the
contradiction to build up.

%\subsection{Proof}
%\label{subsec:conv}

The steps in our proof are the following, illustrated in
Fig.~\ref{fig.converseIllustration}.
\begin{enumerate}
\item FANO: Apply Fano's inequality to each $\rvw_k$ individually.  A
  decision on $\rvw_k$ must be made at time $\Tk = k + T - 1$.
\item GENIE: Condition the decoding of $\rvw_k$ on all previous
  messages $\rvw_0^{k-1}$.  This can only help the decoder (thereby
  increasing the DMT) because the decoder knows exactly the value of
  all earlier messages.  This step can be thought of as a
  genie-helper.
\item SUFFIX OUTAGE: Next we condition on the event that the {\em
  suffix} of the codeword is in outage.  By suffix we mean the symbols
  transmitted in blocks $k, k+1, \ldots, \Tk = T + k-1$.  We bound
  this event using the standard DMT analysis.
\item COMBINE EVENTS: Finally, using standard information
  manipulations, we combine events up to message $\rvw_{N-T+1}$ for
  some large $N$ (to be determined).
\item CONTRADICTION: Finally, using the statistical description of the
  channel law, we find that for any $\delta > 0$ we can identify a
  finite $N$ sufficiently large such that a contradiction arises and
  this demonstrates that a DMT of $T \cdot d_1(r-\delta)$ is not
  achievable.
\end{enumerate}

\begin{figure*}
\centerline{\includegraphics[scale=0.55]{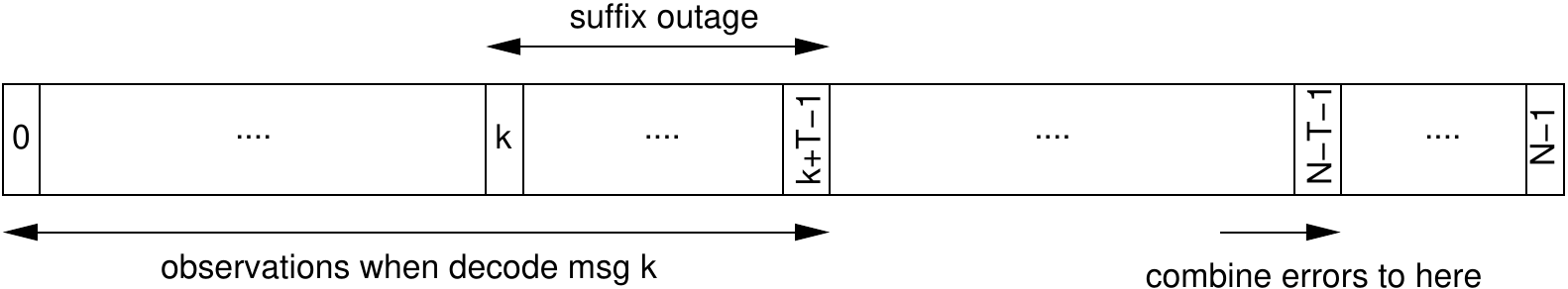}}
\caption{ In the converse, we consider a total of $N$ coherence blocks
  and $N-T$ messages. For each message $k$, we consider the event that
  the coherence blocks $k,\ldots, k+T-1$ are in outage and lower bound
  the error probability in~\eqref{eq.condFanos}. We then combine the
  error probabilities associated with all the messages to obtain a
  lower bound on the maximum error. }
\label{fig.converseIllustration}
\end{figure*}

%\vspace{-1em}

Following the approach outlined above, in our first step we apply
Fano's inequality~\cite[Chapter~2]{coverThomas} to lower bound the
error probability associated with message $\rvw_k$.  To do this, we
define $\cE_k$ to be the event that $\hat{\rvw}_k \neq \rvw_k$ and
note that $\Pr[\mbox{error}] = \sup_{k \ge 0}\Pr[\cE_k]$.  We start by
indexing the error events pointwise in the possible channel
realizations. Strictly speaking, the summation over the channel gains
must be an integral, since the channel gains are continuous valued. We
however use summations, so that the expressions are easier to
follow. All the steps in this section easily follow when the
summations are replaced by corresponding multi-integrals.
\begin{equation}
\Pr[\cE_k]  = \sum_{\bH_0^\Tk} \Pr[\rvbH_0^{\Tk} = \bH_0^{\Tk}] \Pr[\cE_k| \rvbH_0^{\Tk} = \bH_0^{\Tk}]. \label{eq.pointwise}
\end{equation}
Message $\rvw_k$ needs to be decoded at time $\Tk = k + T - 1$.  The
observations accumulated to that time are the channel outputs
$\rvbY_0^{\Tk}$.  Recognizing that the log-cardinality of the message
set from which $\rvw_k$ is chosen is $Mr \log \rho$, we apply Fano's
inequality to each channel realization
to get
{\allowdisplaybreaks{\begin{align*}
&\Pr[\cE_k| \rvbH_0^{\Tk} = \bH_0^{\Tk}]  \ge \frac{- 1  +
  H(\rvw_k | \rvbY_0^{\Tk}, \rvbH_0^{\Tk} = \bH_0^{\Tk}) }{M r \log
  \rho} \\ &= \frac{- 1 + H(\rvw_k) - H(\rvw_k) +
  H(\rvw_k | \rvbY_0^{\Tk}, \rvbH_0^{\Tk}= \bH_0^{\Tk}) }{M r \log \rho}\\ & \ge 1
- \frac{1}{Mr\log\rho} -\frac{H(\rvw_k| \rvw_0^{k-1},
  \rvbH_0^\Tk= \bH_0^{\Tk})}{Mr\log\rho} +\notag\\ &\qquad\frac{H(\rvw_k| \rvbY_0^{\Tk},
  \rvbH_0^{\Tk}= \bH_0^{\Tk})}{Mr\log\rho},
\end{align*}}}
where the latter inequality follows since $H(\rvw_k) =
H(\rvw_k|\rvw_0^{k-1}, \rvbH_0^{\Tk}= \bH_0^{\Tk})$.

The second step in our proof is the genie-aided step.  We condition
the last term in the above on all previous messages yielding the further lower bound:
{\allowdisplaybreaks{\begin{align}
&\Pr[\cE_k | \rvbH_0^\Tk = \bH_0^{\Tk}] \\& \ge 1 -
  \frac{1}{Mr\log\rho} 
-\frac{H(\rvw_k| \rvw_0^{k-1}, \rvbH_0^\Tk= \bH_0^{\Tk})}{Mr\log\rho}\notag\\
&\qquad+\frac{H(\rvw_k| \rvw_0^{k-1}, \rvbY_0^{\Tk}, \rvbH_0^{\Tk}= \bH_0^{\Tk})}{Mr\log\rho}\notag\\
& = 1 -
  \frac{1}{Mr\log\rho} - \frac{I(\rvw_k;  \rvbY_0^{\Tk}| \rvw_0^{k-1}, \rvbH_0^{\Tk}= \bH_0^{\Tk})}{Mr\log\rho}\\
& = 1 -
  \frac{1}{Mr\log\rho} - \frac{I(\rvw_k;  \rvbY_k^{\Tk}| \rvw_0^{k-1}, \rvbH_k^{\Tk}= \bH_k^{\Tk})}{Mr\log\rho}
\label{eq.fanos}
%\\  &=1 - \frac{1}{Mr\log\rho} - \frac{I(\rvw_k;  \rvbY_k^{\Tk}| \rvw_0^{k-1}, \rvbH_k^{\Tk})}{Mr\log\rho}\label{eq:Ek-bnd}
\end{align}}}
To get the last equality we note the following conditions.  
First, since
message and channel realizations are independent $H(\rvw_k|
\rvw_0^{k-1}, \rvbH_0^{\Tk} = \bH_0^{\Tk}) = H(\rvw_k|
\rvw_0^{k-1}, \rvbH_k^{\Tk} = \bH_k^{\Tk})$.
Second, we note the following Markov relationship: $\rvw_k
\leftrightarrow \rvw_0^{k-1}, \rvbY_k^\Tk, \rvbH_k^\Tk \leftrightarrow
\rvbY_0^{k-1}, \rvbH_0^{k-1}$.  This relation holds due to the causal
nature of the encoder and the i.i.d.\ nature of the channel.  In
particular, note that causal encoding means that the channel inputs
$\rvbX_k^\Tk$ are a function of $\rvw_k$ and $\rvw_0^{k-1}$ while past
channel inputs are a function only of $\rvw_0^{k-1}$.  Thus, since the
channel is memoryless the past channel output and state information
$(\rvbY_0^{k-1}, \rvbH_0^{k-1})$ provides no information about
$\rvw_k$ that $(\rvw_0^{k-1}, \rvbY_k^\Tk, \rvbH_k^\Tk)$ does not
provide.

In the third step we condition on the suffix being in outage.  In
particular, define the single-block outage set
\begin{align}
\mathcal{H}_\delta = \big\{ \bH : I(\rvbx; \rvby | \rvbH = \bH) \leq (r - \delta) \log \rho \big\} \label{eq.defOutageSet}
\end{align}
By the classic outage analysis~\cite{ozarowShamaiWyner:94,
  telatarMimo}, which underlies the DMT of
Theorem~\ref{thm:DMT_Zheng}, we know that
\begin{equation}\label{eq:P_del}
P_\del = \Pr[\rvbH \in \cH_\del] \doteq \rho^{-d_1(r-\del)}
\end{equation} 
where $d_1(\cdot)$ is the DMT specified in Theorem~\ref{thm:DMT_Zheng}
and the exponential equality is at high SNR.  By ``suffix outage'' we
mean that $\rvbH_j$ is in outage for {\em every} block $j = k, \ldots,
\Tk$, in other words $\cap_{j=k}^{\Tk} (\rvbH_j \in \cH_\del)$.  Using
$\cH_\del^T$ to denote the $T-$fold Cartesian product of the set
$\cH_\del$, and recalling that the channel gains are sampled in an
i.i.d.\ fashion across blocks, we have
\begin{equation}
\Pr\left[\cap_{j=k}^{\Tk} (\rvbH_j \in \cH_\del)\right] = 
\Pr\left[\rvbH_k^{\Tk}\in \cH_\del^T\right] = (P_\del)^T \doteq \rho^{-T d_1(r-\del)}. \label{eq:probH}
\end{equation} 

We next incorporate the effect of outage into our lower bound.  In Appendix~\ref{app:outage}
we show that
\begin{multline}
\Pr[\cE_k]  \ge \Pr[\rvbH_k^\Tk \in \cH_\del^T] \times \\ \left(1 -
  \frac{1}{Mr\log\rho} - \frac{I(\rvw_k;  \rvbY_k^{\Tk}| \rvw_0^{k-1}, \rvbH_k^{\Tk}, \rvbH_k^\Tk \in \cH_\del^T)}{Mr\log\rho} \right).\label{eq.condFanos}
%& = \Pr[\rvbH_k^{\Tk} \in \cH_\del^T] \Pr[\cE_k | \rvbH_k^{\Tk} \in \cH_\del^T].
\end{multline}
where the expression $\rvbH_k^\Tk \in \cH_\del^T$ in the conditioning indicates that the sequence $\rvbH_k^\Tk$ belongs to the outage set $\cH_\del^T$. %which is defined in~\eqref{eq.defOutageSet}.

Since all the terms in the mutual information expression in~\eqref{eq.condFanos} are independent of the channel gains: $\{\rvbH_0^{k-1},\rvbH_{\Tk+1}^{N-1}\}$ we can express
\begin{align}
&I(\rvw_k;  \rvbY_k^{\Tk}| \rvw_0^{k-1}, \rvbH_k^{\Tk}, \rvbH_k^\Tk \in \cH_\del^T)\notag \\&= I(\rvw_k;  \rvbY_k^{\Tk}| \rvw_0^{k-1}, \rvbH_0^{N-1}, \rvbH_0^{N-1} \in \cH_\del^N ) \label{eq.mutInformA}\\
&\le I(\rvw_k;  \rvbY_0^{N-1}| \rvw_0^{k-1}, \rvbH_0^{N-1}, \rvbH_0^{N-1} \in \cH_\del^N). \label{eq.mutInformB}
\end{align}
where the last step follows from the fact that the mutual information is non-negative.
We will see that the final loosening in~(\ref{eq.mutInformB}) doesn't
weaken our bound for two reasons.  The first is because we study the
max error probability.  The second is because the information about
the messages embedded in later channel uses, i.e.,
$\rvbY_{\Tk+1}^{N-1}$, must be used to decrease the entropy of later
messages.  It cannot be focused exclusively on reducing the
uncertainty of $\rvw_k$ without detrimental effects on the ability to
estimate later messages.  This coupling of errors across time is what we call the
{\emph{outage amplification}} effect.

In the fourth step we combine events.  Substituting~(\ref{eq:probH})
and (\ref{eq.mutInformB})
into~(\ref{eq.condFanos}) we find
\begin{multline}
\Pr[\cE_k] \geq P_\del^T \times \\ \left[1 -
  \frac{1}{Mr\log\rho} - \frac{I(\rvw_k;  \rvbY_0^{N-1}| \rvw_0^{k-1}, \rvbH_0^{N-1}, \rvbH_0^{N-1} \in \cH_\del^N)}{Mr\log\rho}\right].
\end{multline}
And, since the max error is at least as large as the average error, we can arrive at 
\begin{multline} \max_{0 \le k\le N-T-1} \Pr[\cE_k]  \ge
P_\del^T\times \\\Bigg[1  -
  \frac{1}{Mr\log\rho} -
  \frac{I(\rvbX_0^{N-1};\rvbY_0^{N-1}|\rvbH_0^{N-1}, \rvbH_0^{N-1}\in
    \cH_\del^N)}{(N-T)Mr\log\rho}\Bigg]\label{eq:DP}
\end{multline}
as shown in the steps between~\eqref{eq:DP0}--\eqref{eq:DP2}.
\begin{figure*}
\begin{align} &\max_{0 \le k\le N-T-1} \Pr[\cE_k]  \ge
  \frac{1}{N-T}\sum_{k=0}^{N-T-1}\Pr[\cE_k]\label{eq:DP0}\\ 
&\ge P_\del^T\Bigg[1  -
  \frac{1}{Mr\log\rho} \!-\! \frac{\sum_{k=0}^{N-T-1}\!
      I(\!\rvw_k;\rvbY_0^{N-1} |\rvw_0^{k-1},\rvbH_0^{N-1},\rvbH_0^{N-1}\in
      \cH_\del^N\!)}{(N-T)Mr\log\rho}\Bigg]\notag\\
&=P_\del^T\Bigg[1  -
  \frac{1}{Mr\log\rho} -
    \frac{I(\rvw_0^{N-T-1};\rvbY_0^{N-1}|\rvbH_0^{N-1}, \rvbH_0^{N-1}\in
      \cH_\del^N)}{(N-T)Mr\log\rho}\Bigg]\notag\\
&\ge P_\del^T\Bigg[1  -
  \frac{1}{Mr\log\rho} -
  \frac{I(\rvbX_0^{N-1};\rvbY_0^{N-1}|\rvbH_0^{N-1}, \rvbH_0^{N-1}\in
    \cH_\del^N)}{(N-T)Mr\log\rho}\Bigg]\label{eq:DP2}
\end{align}
\end{figure*}
where~\eqref{eq:DP} follows from the data processing inequality since
$\rvw_0^{N-T-1} \rightarrow \rvbX_0^{N-1}\rightarrow \rvbY_0^{N-1}$
holds regardless of the channel realization.

In the final step we apply the channel statistics to get a
contradiction.  In particular, since fading across different blocks is
independent we can break the mutual information term in~(\ref{eq:DP})
into a simple sum
\begin{align}
&\max_{0 \le k\le N-T-1} \Pr[\cE_k] \notag \\ \ge &P_\del^T\Bigg[1 -
  \frac{1}{Mr\log\rho}-
  \frac{\sum_{j=0}^{N-1}I(\rvbX_j;\rvbY_j|\rvbH_j, \rvbH_j\in
    \cH_\del)}{(N-T)Mr\log\rho}\Bigg]\\ &\ge P_\del^T\Bigg[1 -
  \frac{1}{Mr\log\rho} -
  \frac{NM(r-\delta)\log\rho}{(N-T)Mr\log\rho}\Bigg]\label{eq:ChGain}\\ &\doteq
  \Bigg[1-
  \frac{1}{Mr\log\rho} - \frac{N(r-\delta)}{(N-T)r}\Bigg]\rho^{-Td_1(r-\del)}\label{eq:ChGain2}
\end{align}
where~\eqref{eq:ChGain} follows from the
definition~\eqref{eq.defOutageSet} of $\cH_\del$ and we recall that
there are $M$ channel uses in each coherence interval.

To see the contradiction we assume high SNR, so that the second term
vanishes.  Then, for any $\del >0$, by selecting $N > T\frac{r}{\del}$
the term $\big[1 - \frac{N(r-\del)}{(N-T)r}\big]$ is strictly
positive.  Since $\del>0$ is arbitrary, it follows that a diversity
order greater than $Td_1(r)$ cannot be achieved.

We observe that the $N$ required to realize a contradiction is
inversely proportional to $\del$.  This means that if you operate your
system to exceed the DMT by a very small amount it will take some time
for a contradiction to build up.    A coding scheme
can be designed so that early message can borrow channel resources
from later message to ensure their reliability.  But, eventually, the
borrowing builds up and later generations cannot meet their
obligations.  The parameter $N$ indexes the generation that runs into
difficulty.

%Since the second term vanishes as the coherence period $M\rightarrow
%\infty$, we ignore it in our analysis. To bound the remaining terms we let
%\begin{equation}
%\cH_\del = \left\{\bH: \log\det\left(\bI + \frac{\rho}{M}\bH
%\bH^\dagger\right) \le (r-\del)\log\rho \right\}\label{eq:Hdef}
%\end{equation}

%The average error probability
%\begin{align*}
%${\Pr(\cE_k) = E_{\rvbH^T}\left[\Pr(\cE_k; \rvbH_k^{k+T-1})\right]},$ can be lower bounded s follows.
%\end{align*}
%\begin{align}
%\Pr(\cE_k) &\ge \Pr(\rvbH_k^{k+T-1}\in \cH_\del^T)\cdot
%  \Bigg(1-\frac{I(\rvw_k;\rvbY_k^{k+T-1}|\rvw_0^{k-1},
%    \rvbH_k^{k+T-1}\in \cH_\del^T)}{Mr\log\rho}\Bigg)\label{eq:E1-bnd}\\ &=
%  \left(P_\del\right)^T \Bigg(1-\frac{I(\rvw_k;\rvbY_k^{k+T-1}|\rvw_0^{k-1},
%    \rvbH_k^{k+T-1}\in
%    \cH_\del^T)}{Mr\log\rho}\Bigg)\notag\\ &=  \left(P_\del\right)^T
%  \Bigg(1-\frac{I(\rvw_k;\rvbY_k^{k+T-1}|\rvw_0^{k-1},
%    \rvbH_0^{N-1}\in
%    \cH_\del^N)}{Mr\log\rho}\Bigg)\label{eq:t1}\\ &\ge   \left(P_\del\right)^T
%  \Bigg(1-\frac{I(\rvw_k;\rvbY_0^{N-1}|\rvw_0^{k-1}, \rvbH_0^{N-1}\in
%    \cH_\del^N)}{Mr\log\rho}\Bigg)\label{eq:EkBound}
%\end{align}
%where~\eqref{eq:t1} follows from the fact that the channel gains
%$(\rvbH_0^{k-1},\rvbH_{k+T}^{N-1})$ are independent of
%$(\rvw_0^k,\rvbY_k^{k+T-1}, \rvbH_k^{k+T-1})$. 

\section{Coding Theorem}
\label{sec:Coding}

We present two approaches for achieving the DMT stated in
Theorem~\ref{thm:streaming_DMT}. As mentioned in the introduction, the
first approach is based on interleaving the last $T$ messages across
coherence block while the second approach is based on a
delay-universal tree code construction.

\subsection{Interleaving Scheme}
\label{sec:Coding-Interleaving}

We show that a simple interleaving based scheme suffices to achieve
the DMT stated in Theorem~\ref{thm:streaming_DMT}. Our codebook $\cC$
maps each message $\rvw_k \in \cI_M \defeq \{1,2,\ldots,
2^{Mr\log\rho}\}$ to $T$ codewords $\left\{\rvbX_0(\rvw_k),
\rvbX_1(\rvw_k),\ldots, \rvbX_{T-1}(\rvw_k)\right\}$ where each
$\rvbX_j \in \mathbb{C}^{\Nt \times \frac{M}{T}}$. Thus the overall
code is a Cartesian product: $\cC = \cC_0 \times \cC_1 \ldots
\times\cC_{T-1}$, where $\rvbX_k \in \cC_k$.  We will assume that each
codebook $\cC_j$ is sampled i.i.d.\ according to a complex normal
$\CN(0, \frac{\rho}{\Nt})$ distribution\footnote{We note however that
  any space-time code that achieves the DMT for independent parallel
  MIMO fading channels can be used for the sub-codebooks
  $\cC_0,\ldots, \cC_{T-1}$. In particular the non-vanishing
  determinant (NVD) code in~\cite{yang} can be used for these
  sub-codebooks instead of the random Gaussian codebook.}

For transmission of each message, we assume that each coherence block
of length $M$ is further divided into $T$ sub-blocks of length
$\frac{M}{T}$, as indicated in Fig.~\ref{fig:interleaving}. Suppose
that $\cI_{k,0},\ldots, \cI_{k,T-1}$ denote these intervals.  The
codeword $\rvbX_0(\rvw_k)$ is transmitted in the first sub-block
$\cI_{k,0}$ of coherence block $k$.  The codeword $\rvbX_1(\rvw_k)$ is
transmitted in the sub-block $\cI_{k+1,1}$ of coherence block ${k+1}$
and likewise $\rvbX_j(\rvw_k)$ is transmitted in the $j$-th sub-block,
$\cI_{k,j}$, of coherence block $k+j$. The corresponding output
sequences associated with message $\rvw_k$ are denoted by:
\begin{align}
\rvbY_{k,j} = \rvbH_{k+j} \rvbX_{j}(\rvw_k) + \rvbZ_{k,j}, \quad j=0,\ldots, T-1.
\end{align} 
The decoder finds the message $\hat{\rvw}_k$ such that for each $j \in
\{0,\ldots, T-1\}$ the sequence pair
  $(\rvbX_j(\hat{\rvw}_k), \rvbY_{k,j})$ is jointly weakly
  typical~\cite{coverThomas}. The shaded boxes in
Fig.~\ref{fig:interleaving} denote the intervals used for the decoding
of $\rvw_k$. The outage event at the decoder, associated with
$\hat{\rvw}_k$, is given by:
\begin{align}
\bigg\{\frac{1}{T}\sum_{j=k}^{k+T-1} C_j(\rho) \le r\log \rho\bigg\}
\label{eq:outage}\end{align}
where $C_j(\rho) = \log\det\left(I + \frac{\rho}{\Nt}\rvbH_j
\rvbH_j^\dagger\right)$. Since~\eqref{eq:outage} precisely corresponds
to the outage event of a quasi-static parallel MIMO fading channel,
with $T$ channels and a multiplexing gain of $T\cdot r,$ the
achievability of the DMT in Theorem~\ref{thm:streaming_DMT} follows
from Corollary~\ref{corol:parallel}.

\begin{figure*}
\begin{center}
\includegraphics[scale=0.5]{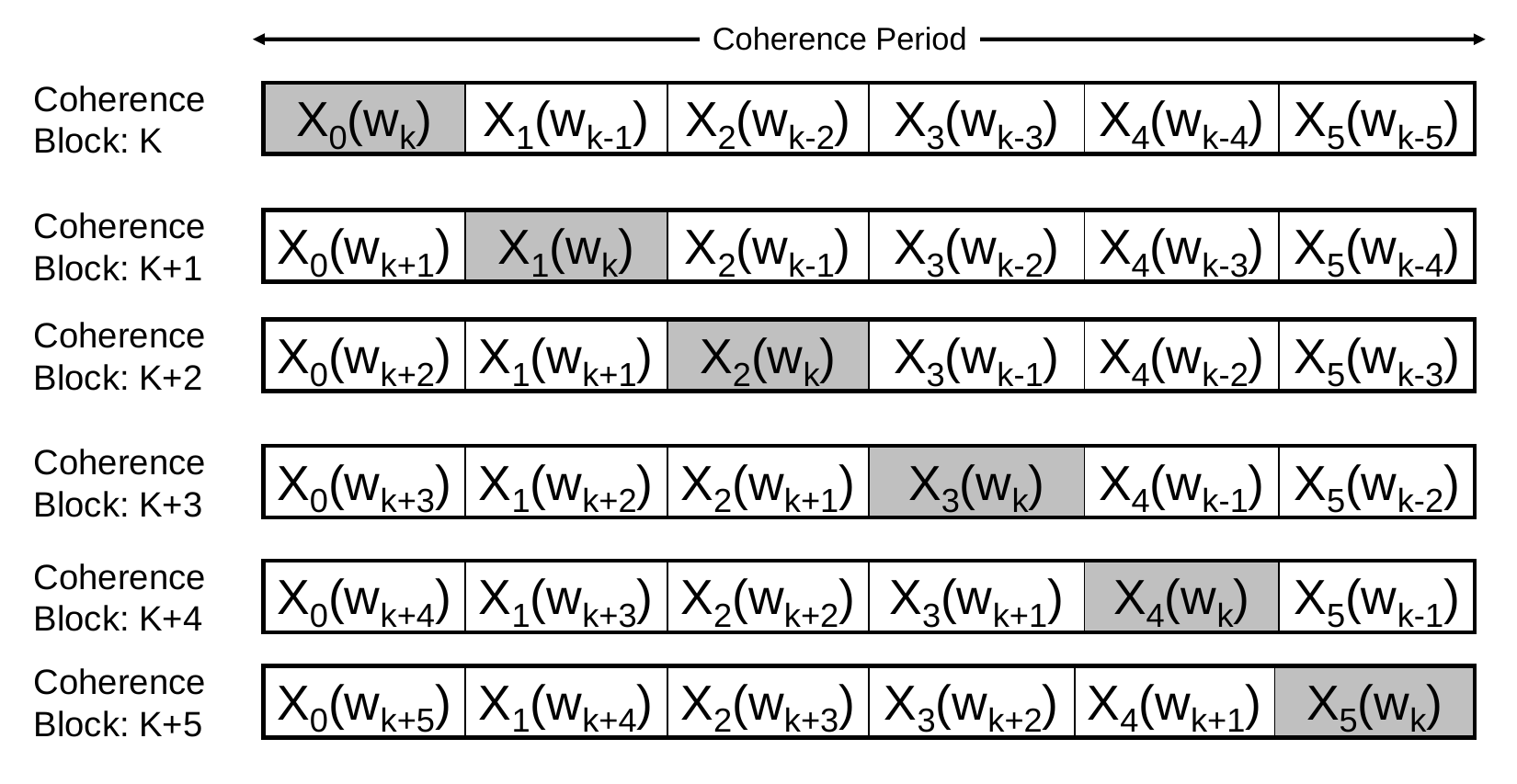}
\caption{Interleaving based coding scheme for $T=6$. Each coherence block is divided into $T$ sub-intervals and each sub-interval is dedicated to transmission of one message. The transmission of message $\rvw_k$ spans coherence blocks $k,k+1,\ldots, k+T-1$ using codewords of $\rvbX_0(\rvw_k),\ldots,\rvbX_{T-1}(\rvw_k)$ as shown by the shaded blocks.}
\label{fig:interleaving}
\end{center}
\end{figure*}

\subsection{Sequential Tree Codes}
\label{sec:Coding-Anytime}

Our second scheme is based on a sequential tree code
construction. This approach has the advantage that the encoder does
not require the knowledge of $T$. The delay constraint only needs to
be revealed to the decoder, and yet the optimal DMT is attained.

Our proposed construction consists of a sequence of codebooks
$\left\{\cC_0,\cC_1,\ldots, \cC_k,\ldots\right\}$, where $\cC_k$ is
the codebook to be used in coherence block $k$ when messages
$(\rvw_0,\ldots, \rvw_k)$ are revealed to the encoder.  Codebook
$\cC_k$ consists of a total of $2^{MR(k+1)}$ codeword sequences, with
one codeword for each element in the set:
\begin{equation}
\label{eq:msg_tuple}
\cI_M^{k+1} = \left\{~(\svw_0,\ldots, \svw_k): \svw_0 \in \cI_M,
\ldots, \svw_k \in \cI_M\right\}.
\end{equation} 
where $\cI_M \defeq \{1,2,\ldots, 2^{MR}\}$. All codewords are sampled
i.i.d.\ from $\CN\left(0,\frac{\rho}{\Nt}\right)$ and are revealed to
both the encoder and the decoder in advance\footnote{We will make the
  assumption that the communication terminates after a sufficiently
  large but fixed number of coherence blocks.}.  In coherence block
$k$, the encoder maps $\rvw_0,\ldots, \rvw_k$ to the codeword
$\rvbX_k(\rvw_0^k) \in \mbC^{\Nt \times M}$ in $\cC_k$, and transmits
it over $M$ channel uses. The entire transmitted sequence up to and
including block $k$ is denoted
by \begin{multline}\label{eq:transmit_sequence} \rvbX_0^k(\rvw_0^k)
  \defeq \left\{ \rvbX_0(\rvw_0),\rvbX_1(\rvw_0^1),\ldots,
  \rvbX_k(\rvw_0^k) \right\}, \\ \rvbX_0^k(\rvw_0^k) \in \mbC^{\Nt
    \times (k+1)M}
\end{multline}

The decoder uses a sequential, decision-directed decoding rule.  We
focus on the decoding of message $\rvw_k$ at the end of coherence
block ${\Tk = k+T-1}$, which corresponds to the deadline of message
$\rvw_k$. The decoder considers the entire received sequence
$\rvbY_0^{\Tk} = (\rvbY_0,\ldots, \rvbY_{\Tk})$ and computes a fresh
estimate of all the messages up to time $k$ in ${k+1}$ steps as
follows. In the first step, the decoder searches over all message
sequences $\hvw_0^{\Tk}$ such that the pair
$(\rvbX_0^{\Tk}(\hvw_0^{\Tk}), \rvbY_{0}^{\Tk})$ is jointly
typical. If each such message sequence has a unique prefix, say
$\bvw_0$, then $\bvw_0$ is selected as the message in block
$0$. Otherwise an error is declared. Once the message $\bvw_0$ is
fixed in the first step, the decoder then proceeds to the second
step. It searches for the message sequences $\hvw_{1}^{\Tk}$ such that
the pair $(\rvbX_1^\Tk(\bvw_0, \hvw_{1}^\Tk), \rvbY_1^\Tk)$ is jointly
typical. If each such message sequence has a unique prefix, say
$\bvw_1$ then it is selected as the message in block $1$. Otherwise an
error is declared. Once the message $\bvw_1$ is fixed the decoder
proceeds sequentially, producing  $\bar{\rvw}_2,\ldots, \bar{\rvw}_k$.
In determining $\bar{\rvw}_l$, with $l \leq k$, the decoder fixes
$\bar{\rvw}_0^{l-1}$ and searches for a sequence of messages
$\hvw_l^{\Tk}$ such that the corresponding transmit sequence
$\rvbX_0^{\Tk}(\bar{\rvw}_0^{l-1}, \hvw_l^{\Tk})$ has the property
that the sub-sequence between $l$ to $\Tk$ (the suffix) satisfies
\begin{equation} 
(\rvbX_l^{\Tk}(\bar{\rvw}_0^{l-1}, \hvw_l^{\Tk}), \rvbY_{l}^{\Tk}) \in
  \cT_{l,\Tk},\label{eq:decoding_rule}
\end{equation} 
where the set $\cT_{l,l'}$ is the set of all jointly typical
sequences~\cite{coverThomas},
\begin{multline} \cT_{l,l'}
  = \bigg\{(\svbX_l^{l'}, \svbY_l^{l'} ): \svbX_l^{l'} \in \cT(p_{\rvbX_l^{l'}}),
  \svbY_l^{l'} \in \cT(p_{\rvbY_l^{l'}}), \bigg|  \\ \frac{\sum_{k=l}^{l'} [-\log p_{\rvbX_k,
    \rvbY_k}(\svbX_k, \svbY_k) - h(p_{\rvbX_k,
    \rvbY_k})]}{M(l'-l+1)}\bigg| \le \eps \bigg\}.
\label{eq:jointly_typical}
\end{multline}
In~(\ref{eq:jointly_typical}) $\cT(p_{\rvbX_l^{l'}})$ and
$\cT(p_{\rvbY_l^{l'}})$ denotes the set of typical $\{\rvbX_l^{l'}\}$
and $\{\rvbY_l^{l'}\}$ sequences respectively and $h(p_{\rvbX_k,
  \rvbY_k})$ denotes the differential entropy of jointly Gaussian
random variables.
%These sets
%are determined by the input distribution, the noise statistics, and
%the known channel matrices $\{\rvbH_k\}$.

\begin{figure*}
\centering
\includegraphics[scale=0.5]{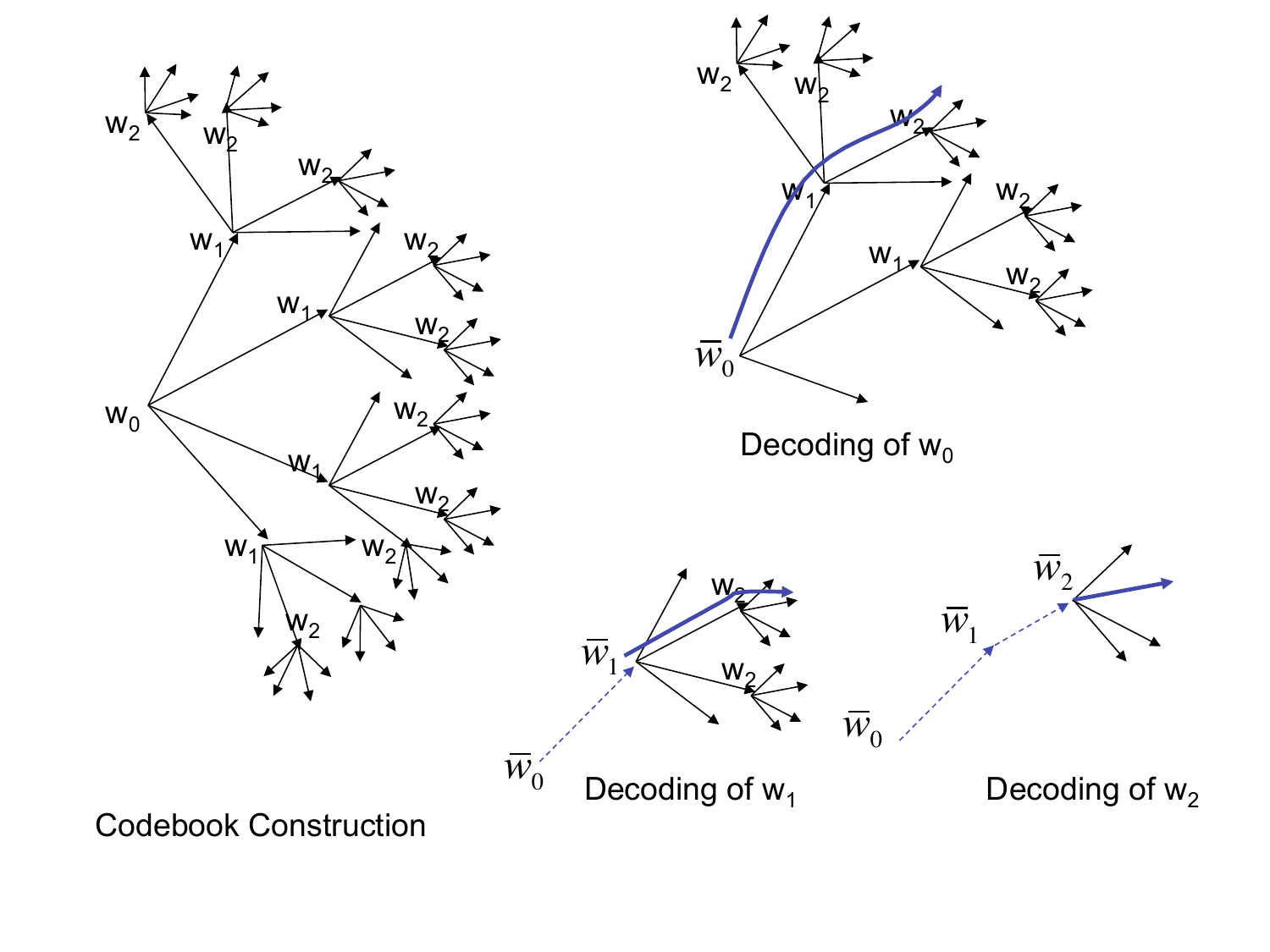}
\caption{ The left  figure illustrates the tree-codebook.  The message $\rvw_0$ is mapped to one of
  $2^{MR}$ codewords in the first level, the message pair
  $(\rvw_0,\rvw_1)$ is mapped to one of $2^{2MR}$ codewords in the
  second level etc., While decoding $\rvw_k$ the decoder  starts at the root of the tree. It first finds all possible transmit
  paths of depth ${k+T-1}$ in the tree, typical with the received
  sequence. If a unique prefix codeword $\bar{\rvw}_0$ is determined
then  the decoder moves along the path of $\bar{\rvw}_0$ and finds all possible
  codewords from level $1$ to ${k+T-1}$ that are typical with the
  received codeword. A unique message $\bar{\rvw}_1$ is determined if there
  is a unique prefix codeword in this level. This process continues
  till level ${k}$ is reached and $\bar{\rvw}_k$ is determined. }
 \label{fig:treeCode}
\end{figure*}

If the list of all message sequences $\hvw_l^{\Tk}$ that
satisfy~\eqref{eq:decoding_rule} have a unique prefix $\bar{\rvw}_l$
then we concatenate $\bar{\rvw}_l$ with $\bar{\rvw}_0^{l-1}$ to get
$\bar{\rvw}_0^l$, otherwise an error is declared.  When the process
continues to step $k+1$ without declaring an error,
$\bar{\rvw}_k$ is declared to be the output message estimate, i.e.,  $\hat{\rvw}_k = \bar{\rvw}_k$.

Fig.~\ref{fig:treeCode} illustrates the codebook construction and the
proposed sequential decoding rule.  The figure on the left hand side
illustrates the sequential tree code.  The right figures illustrate
the sequential decoding of $\rvw_0$, $\rvw_1$ and $\rvw_2$. When
decoding $\rvw_0$, we consider all possible paths in the tree typical
with the received sequence. If all such paths lead to a unique prefix
$\bvw_0$, we declare this to be the message. Otherwise an error is
declared. Once $\bvw_0$ is fixed, we move along the path of $\bvw_0$
in the tree. Thereafter we search for all paths in the tree from level
$1$ to $k+T-1$ that are typical with the received sequence. This
process continues until level $k$ is reached and $\bar{\rvw}_k$ is
determined.

\begin{remark}
Our decoder is a decision directed decoder. In estimating $\bvw_0^k$,
it first estimates $\bvw_0$ based on $\rvbY_0^{\Tk}$.  It next makes a
conditional estimate of $\bvw_1$ based on $\rvbY_1^{\Tk}$ with
$\bvw_0$ fixed, and continues along in ${k+1}$ steps.  One may be
tempted to try a simpler decoding scheme that avoids the ${k+1}$ steps
and directly search for a unique prefix $\hvw_0^k$ such that the
resulting transmit sequence $\rvbX_0^{\Tk}$ is jointly typical with
the received sequence $\rvbY_0^{\Tk}$ i.e.,
\begin{align}\bigg\{\bigg| \frac{\sum_{k=0}^{\Tk} [-\log p_{\rvbX_k,
    \rvbY_k}(\svbX_k, \svbY_k) - h(p_{\rvbX_k,
    \rvbY_k})]}{M(k+T) }\bigg| \le \eps \bigg\}.\label{eq:typ-test}\end{align}

Such an approach will not guarantee the recovery of the true $\rvw_k$
with high probability. This is because for $k\gg 1$ the contribution
of the terms before $\hat{\rvw}_k$ will dominate. Even when
$\hat{\rvw}_k \neq \rvw_k$ but ${\hat{\rvw}_0^{k-1}=\rvw_0^{k-1}},$
the pair $(\hat{\rvbX}^{\Tk}, \rvbY^{\Tk})$ will in general
satisfy~\eqref{eq:typ-test} as for $k \gg 1$ the contribution of the
suffix associated with $\hat{\rvw}_k$ will be negligible.  In other
words, our proposed decision directed decoder
in~\eqref{eq:decoding_rule} fixes the older messages and guarantees
that when decoding $\rvw_k$ we do not include the bias introduced by
$\rvw_0^{k-1}$ in~\eqref{eq:typ-test}.
\end{remark}

\subsubsection{Analysis of error probability}
We show that for any $\del > 0$ and $0 < r < \min(\Nr , \Nt)$, the
error probability averaged over the ensemble of codebooks satisfies
$\Pr(\hat{\rvw}_k \neq \rvw_k) \stackrel{\cdot}{\le} \rho^{-T\cdot
  d(r)}$. In our analysis, we exploit symmetry in the code
construction, as well as the encoding and decoding functions.  To lay out the analysis 
assume, without loss of generality, that a particular message sequence
$\rvw_0^k = \svw_0^k$ has been transmitted. Define the events\footnote{All
  the error events $\cE_l$ are defined for the decoder at time
  $T+k-1$. However we suppress this dependence to keep the notation
  compact.}
\begin{multline}
\cE_l = \bigg\{ \svw_0^l: (\svw_0,\ldots, \svw_{l-1})= \\(w_0,\ldots,
w_{l-1}), \; \bvw_l \neq w_l \bigg\},  \qquad 0 \le l \le
k \label{eq:Pe_def}
\end{multline} 
and note that
\begin{equation} 
 \Pr\left\{\bvw_k \neq w_k \right\} \le \sum_{l=0}^k \Pr(\cE_l),
 \label{eq:union_bound}
\end{equation}
where $\cE_l$ corresponds to the event that our proposed decoder fails
in step $l$ of the decoding process. We develop an upper bound on
$\cE_l$ for each $0 \le l \le k$ and substitute these bounds
in~\eqref{eq:union_bound}. We further express ${\cE_l = \cA_l \cup
  \cB_l,} $ where
\begin{equation}
\cA_l = \left\{ (\svbX_l^{\Tk}(w_0^\Tk), \svbY_l^{\Tk}) : (\svbX_l^{\Tk}(w_0^\Tk), \svbY_l^{\Tk}) \notin
\cT_{l,\Tk} \right\}
\label{eq:MsgSeq_Atypical}
\end{equation}
denotes the event that a decoding failure happens because the
transmitted sub-sequence starting from position $l$ fails to be
typical with the received sequence, whereas
\begin{equation}
\cB_l = \bigg\{ \bar{w}_0^{T_k}: \bar{w}_0^{l-1} = w_0^{l-1}, \bar{w}_l
\neq w_l,  (\svbX_l^{\Tk}(\bar{w}_0^\Tk), \svbY_l^{\Tk}) \in
\cT_{l,\Tk} \bigg\}
\label{eq:FalseMsgSeq_Typical}
\end{equation}
denotes the event that the decoding failure happens because a transmit
sequence corresponding to a message sequence with ${\bar{w}_l \neq w_l}$
appears typical with the received sequence.

As shown in the Appendix~\ref{app:PrE_A_proof}, using an appropriate
Chernoff bound we can express,
\begin{equation}
\Pr(\cA_l) \le 2^{-M (T_k - l + 1) f(\eps)} = 2^{-M(T+k-l)f(\eps)} \label{eq:PrE_A_Bound}
\end{equation}
where $f(\eps)$ is a function that satisfies $f(\eps)>0$ for each $\eps>0$.

To bound $\Pr(\cB_l)$ we begin by noting that by our code
construction, we are guaranteed that whenever $\bar{w}_l \neq w_l$, the
associated transmitted subsequence $\rvbX_l^\Tk(\bar{w}_0^{\Tk})$ is sampled
independently of $\rvbY_l^\Tk$. Hence from the joint typicality
analysis~\cite{coverThomas}, we have that for any sequence $\bar{w}_0^\Tk$
with ${\bar{w}_l \neq w_l}$
{\allowdisplaybreaks{\begin{align*}
&\Pr\Big((\rvbX_l^{\Tk}(\bar{w}_0^\Tk), \rvbY_l^{\Tk})  \in
\cT_{l,\Tk}~|~\rvbH_l^\Tk = \bH_l^\Tk\Big) \\&\le 2^{-M
  \left(\sum_{j=l}^\Tk I(\rvbx_j; \rvby_j|\rvbH_j=\bH_j) - 3\eps\right)}\\ & =
2^{-M \left(\sum_{j=l}^\Tk C_j(\rho; \bH_j) - 3\eps\right)}
\end{align*}}}
where   
\begin{equation}
C_j(\rho; \bH_j) \defeq \log\det\left(\bI + \frac{\rho}{\Nt} \bH_j
\bH_j^\dagger\right)\label{eq:MIMO_Cap}
\end{equation} 
is the associated mutual information between the input and output in
the $j$-th coherence block when the channel matrix $\rvbH_j = \bH_j$.
Applying the union bound we have that
{\allowdisplaybreaks{\begin{align}
&\Pr(\cB_l~|\rvbH_l^\Tk=\bH_l^{\Tk}) \\ &\le \sum_{\bar{w}_{l}^{\Tk} \in \cI_{M}^{\Tk-l+1}} \Pr\Big((\rvbX_l^{\Tk}(\bar{w}_0^\Tk), \rvbY_l^{\Tk})  \in
\cT_{l,\Tk}~|~\rvbH_l^\Tk=\bH_l^\Tk\Big) \\
&\le \sum_{\bar{w}_{l}^{\Tk} \in \cI_{M}^{\Tk-l+1}} 2^{-M \left(\sum_{j=l}^\Tk C_j(\rho; \bH_j) - 3\eps\right)}\\
&\le \left(2^{M(\Tk-l+1)R} \right)  2^{-M \left(\sum_{j=l}^\Tk C_j(\rho; \bH_j) - 3\eps\right)}\\
&\le 2^{-M \left(\sum_{j=l}^\Tk
    C_j(\rho; \bH_j)-(\Tk-l+1)R- 3\eps\right)}.\label{eq:PrE_B_Bound}
\end{align}}}

To bound $\Pr(\cB_l)$ we define 
{\allowdisplaybreaks{\begin{multline}
\cO_l = \Bigg\{ (\bH_l,\ldots, \bH_{\Tk}): \\ \sum_{i=l}^{\Tk}C_i(\rho; \bH_i)
\le  (k+T-l)r \log\rho +(k-l)\Delta(r)\log \rho + 4\epsilon \log
\rho\Bigg\}
\label{eq:OlDef}\end{multline} }}
where
%\begin{equation}
%\Delta(r) = \frac{(\Nt-r)(\Nr-r)}{2(\Nt+\Nr-2r)}, \;\;  0 \le r <
%\min(\Nt,\Nr).\label{eq:offset_def}
%\end{equation}
\begin{equation}
\Delta(r) = -\frac{d_1(r)}{2d_1'(r)} \label{eq:del-def}
\end{equation}
where we recall that $d_1(r)$ denotes the quasi-static
DMT~\eqref{eq:quasi-static_DMT} of the MIMO fading channel and we use
$d_1'(r)$ to denote its right derivative. Note that $d_1'(r) <0$ for
all $r \in [0, \min(\Nt, \Nr)]$ as the DMT is a decreasing function of
$r$. Thus it follows that $\Delta(r)>0$.

Note that
\begin{align}
\Pr(\cB_l) &\le \Pr(\cB_l~|~\rvbH_l^\Tk \in \cO_l^c) + \Pr(\rvbH_l^\Tk \in\cO_l). 
\label{eq:outage_events}
\end{align}
From~\eqref{eq:OlDef} and~\eqref{eq:PrE_B_Bound} we have \begin{align}
  \Pr(\cB_l~|~\rvbH_l^\Tk \in \cO_l^c) &\le 2^{-M{(k-l)}\Delta(r)\log \rho - M{\eps}\log
    \rho} \\ &= \rho^{-M\eps -
    M(k-l)\Delta(r)}. \label{eq:PrE_B_Bound_Oc}
\end{align}
We next upper bound the second term in~\eqref{eq:outage_events}.  Note
that $\cO_l$ precisely corresponds to the parallel MIMO channel in
Corollary~\ref{corol:parallel} with $L \defeq \Tk - l +1 = k +T-l$
channels, and multiplexing gain  $s={L r + (k-l)\Delta(r)
  +4\eps}$. The associated DMT satisfies:
\begin{align}
L \cdot d_1\left(\frac{s}{L}\right) &= L \cdot d_1\left(r + \frac{k-l}{L}\Delta(r) + \frac{4\eps}{L}\right)\\
&=L \cdot d_1\left(r + \frac{k-l}{L}\Delta(r) \right) + o_\eps(1) \label{eq:d1-cont} \\
&\ge L\left(d_1(r) + d_1'(r) \frac{k-l}{L}\Delta(r)\right) + o_\eps(1) \label{eq:d1-diff}\\
&= Ld_1(r) - \frac{(k-l)}{2}d_1(r) +o_\eps(1) \label{eq:del-sub}\\
&= Td_1(r) + (k-l) d_1(r)-\frac{(k-l)}{2}d_1(r) +o_\eps(1)\label{eq:L-sub}\\
&= T d_1(r) + \frac{(k-l)}{2}d_1(r) + o_\eps(1)
\end{align}
where we use the continuity of $d_1(r)$ in~\eqref{eq:d1-cont} and let $o_\eps(1)$ be a function of $\eps$ that
vanishes as $\eps \rightarrow 0$. We use the convexity of $d_1(r)$ in~\eqref{eq:d1-diff}.
We substitute~\eqref{eq:del-def} for $\Delta(r)$ in~\eqref{eq:del-sub} and substitute $L=T+k-l$ in~\eqref{eq:L-sub}.
Thus we have\begin{align} \Pr(\cO_l) &\stackrel{\cdot}{\le} \rho^{-(T
  d_1(r) + \frac{(k-l)}{2}d_1(r))+ o_\eps(1)}.\label{eq:P_Ol_Bound}
\end{align}
\iffalse
and
\begin{align}
\Pr(\rvbH_l^\Tk  \in \cO) &\le \sum_{l=0}^k\Pr(\cO_k)\\
&\doteq \rho^{-Td_1(r) + o_\eps(1)}\sum_{l=0}^k \rho^{-\frac{(k-l)}{2}d_1(r)} \doteq \rho^{-Td_1(r)}.
\label{eq:PrO_Bound}
\end{align}
\fi 
From~\eqref{eq:outage_events} and
substituting~\eqref{eq:PrE_B_Bound_Oc} and~\eqref{eq:P_Ol_Bound} and
using $\cE_l = \cA_l \cup \cB_l$ we have
\begin{align}
&\Pr(\cE_l) \le \Pr(\cA_l) + \Pr(\cB_l) \\
&\stackrel{\cdot}{\le} 2^{-M (T+k - l) f(\eps)} +  \rho^{-M\eps - M(k-l)\Delta(r)} +  \rho^{-Td_1(r)  -\frac{(k-l)}{2}d_1(r)+o_\eps(1)}.
\label{eq:PrEl_Bound}
\end{align}

From the union bound,
\begin{align}
&\Pr(\cE)  \le    \sum_{l=0}^k \Pr(\cE_l)\\
&\le \sum_{l=0}^k 2^{-M (T+k - l) f(\eps)} + \sum_{l=0}^k  \rho^{-M\eps - M(k-l)\Delta(r)} +  \notag\\ & \sum_{l=0}^k \rho^{-Td_1(r)  -\frac{(k-l)}{2}d_1(r)+o_\eps(1)} .
\label{eq:Pe_bound_two_terms}
\end{align}
We upper bound the first term
in~\eqref{eq:Pe_bound_two_terms} as
\begin{align}
&\sum_{l=0}^k 2^{-M (k+T - l ) f(\eps)} = \sum_{l=0}^k 2^{-M (l+T ) f(\eps)} \\
&\le \sum_{l=0}^\infty 2^{-M (l+T ) f(\eps)} \le 2^{-MT f(\eps)+1},
\end{align}
which vanishes as $M\rightarrow\infty$. By a similar argument we can
upper bound the second term as
\begin{align}
\sum_{l=0}^k  \rho^{-M\eps - M(k-l)\Delta(r)} &= \rho^{-M\eps}\sum_{l=0}^k  \rho^{-Ml\Delta(r)}\\
&\le \rho^{-M\eps}\sum_{l=0}^\infty  \rho^{-Ml\Delta(r)} \le 2\rho^{-M\eps}
%&= \rho^{-M\eps}\frac{1}{1-\rho^{-M\Delta(r)}} \le 
\end{align}
for sufficiently large $\rho $ and $M$ such that $\rho^{-M\Delta(r)} \le \frac{1}{2}$.
In a similar fashion we can upper bound the third term in~\eqref{eq:Pe_bound_two_terms} as
\begin{align}
\sum_{l=0}^k \rho^{-Td_1(r)  -\frac{(k-l)}{2}d_1(r)+o_\eps(1)} &\le \rho^{-Td_1(r)+o_\eps(1)} \sum_{l=0}^k \rho^{-\frac{l}{2}d_1(r)} \\
&\le 2 \rho^{-Td_1(r)+o_\eps(1)}.
\end{align}
From~\eqref{eq:Pe_bound_two_terms} we have that
\begin{align}
\Pr(\cE) \le 2(2^{-MTf(\eps)} + \rho^{-Td_1(r)+o_\eps(1)} + \rho^{-M\eps})
\end{align}
By selecting $M \ge \frac{d_1(r)\log\rho}{f(\eps)}$, we have that $2^{-MTf(\eps)} \le \rho^{-Td_1(r)+o_\eps(1)}$. Finally since $\eps>0$ can be selected as small as required and $o_\eps(1)\rightarrow 0$ as $\eps \rightarrow 0$ we have that $\Pr(\cE) \stackrel{\cdot}{\le} \rho^{-T d_1(r)}$. This completes the error analysis for the sequential tree code.

%Since $\eps> 0$ is arbitrary  we have established that the DMT of $Td_1(r)$ is achievable. 

\section{Multiple Messages per Coherence Block}
\label{sec:multiple}

\begin{figure*}
\begin{center}
\includegraphics[width=\linewidth]{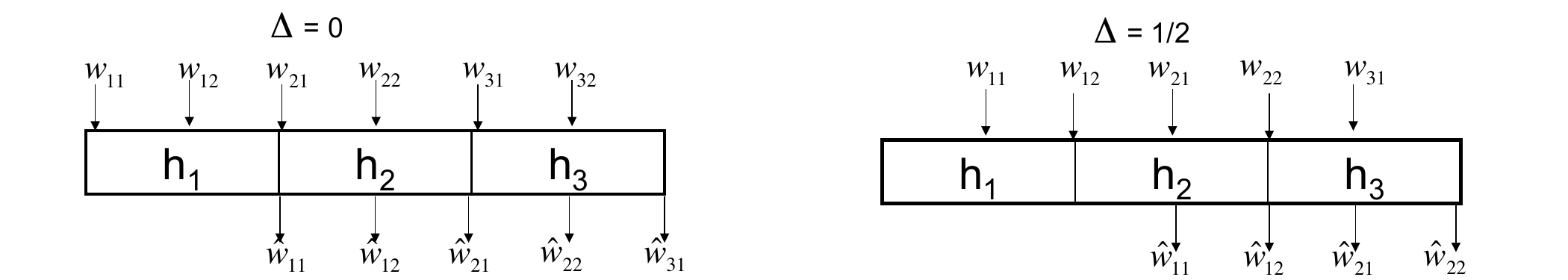}
\end{center}
\caption{Streaming setup with two messages arriving in each coherence
  block. In coherence block $k$ two messages, $\rvw_{k,1}$ and
  $\rvw_{k,2},$ arrive as shown in the figure. We assume that
  $\rvw_{k,1}$ arrives $\Delta \cdot M$ symbols after the start of the
  coherence block $k$, while $\rvw_{k,2}$ arrives $\left(\Delta +
  \frac{1}{2}\right)\cdot~M$ symbols from the start of coherence block
  $k$.  We assume a decoding delay of one coherence block for each
  message as shown. The left figure illustrates the case when
  $\Delta=0,$ while the right figure shows the case when $\Delta =
  \frac{1}{2}$.}
\label{fig:multiple}
\end{figure*}

Our focus so far has been on the case where only one message arrives
in each coherence block. In this section, we consider some special
cases when two messages, say $\rvw_{k,1}$ and $\rvw_{k,2}$ arrive in
each coherence block. Each message must be decoded after $M \cdot T$
channel uses, where $T$ denotes the delay in coherence blocks. In such
a setup, the number of coherence blocks seen by $\rvw_{k,1}$ before
its deadline, will be different from $\rvw_{k,2}$. Thus a simple
interleaving techniques such as is presented in
Section~\ref{sec:Coding-Interleaving} is no longer optimal; more
sophisticated coding techniques that exploit the asymmetry between
$\rvw_{k,1}$ and $\rvw_{k,2}$ will be required.

 Assume that $\rvw_{k,1}$ arrives at time $t_{k,1}= (k-1)M+
 M\cdot\Delta$, while $\rvw_{k,2}$ arrives at time $t_{k,2} = (k-1)M +
 M\left(\Delta + \frac{1}{2}\right)$ where $\Delta \in [0,1/2]$
 denotes the offset relative to the start of the coherence block when
 $\rvw_{k,1}$ arrives. Fig.~\ref{fig:multiple} denotes the streaming
 setup with two messages per block corresponding to $\Delta=0$ and
 $\Delta = 1/2$ respectively.  We obtain the optimal DMT for the SISO
 channel with $T=1$ and $\Delta=0$, which corresponds to the first
 case in Fig.~\ref{fig:multiple}, in Section~\ref{subsec:delta}. By
 the symmetry of the problem, the same result also applies when
 $\Delta=1/2$, illustrated in the second case in
 Fig~\ref{fig:multiple}. In subsection~\ref{subsec:unknown} we show
 that if either $\Delta=0$ or $\Delta=1/2$, but its actual value is
 not known to the encoder, the DMT is strictly smaller.

\subsection{DMT when $\Delta=0$.}
\label{subsec:delta}
We assume that each message $\rvw_{k,i}$ is uniformly distributed in
the set $\cI_M = \{1,2,\ldots, 2^{MR/2}\}$, so that a total of $MR$
information bits arrive in each coherence block of length $M$. We
consider a SISO channel model.  Let $\rvh_k$ denote the channel gain
in coherence block $k$, and denote the corresponding input sequence as
$\rvx_k^M$. We split $\rvx_k^M = \left(\rvx_{k,1}^{M/2},
\rvx_{k,2}^{M/2}\right) \defeq \left(\rvbx_{k,1}, \rvbx_{k,2}\right)$
into two subsequences, each of length $M/2$.  The input sequence
$\rvbx_{k,1}$ can depend on $(\rvw_{k,1},
\rvw_{k-1,1},\rvw_{k-1,2},\ldots),$ while $\rvbx_{k,2}$ can depend on
$(\rvw_{k,1},\rvw_{k,2},\rvw_{k-1,1},\rvw_{k-1,2},\ldots)$. The
received sequence $\rvy_k^M$ is also partitioned into
$\left(\rvy_{k,1}^{M/2},\rvy_{k,2}^{M/2}\right) \defeq
(\rvby_{k,1},\rvby_{k,2})$.  Recall that for the SISO channel, we have
$\rvby_{k,i}=\rvh_k\cdot \rvbx_{k,i} + \rvbz_{k,i}$ where the additive
noise sequence $\rvbz_{k,i}$ is sampled i.i.d.\ from $\CN(0,1)$. We
will assume that each message $\rvw_{k,i}$ must be decoded with a
delay of $T=1$ coherence period. Thus, $\rvw_{k,1}$ must be decoded at
the end of coherence block $k$ whereas $\rvw_{k,2}$ must be decoded in
the middle of coherence block ${k+1}$, as is illustrated in
Fig.~\ref{fig:multiple}.  The following result shows how to exploit
the asymmetry in channel conditions experienced by $\rvw_{k,1}$ and
$\rvw_{k,2}$ to attain a higher DMT than that which can be obtained
through simple interleaving.\footnote{We will drop the subscript
  $d_T(\cdot)$ in this section for the DMT since we fix $T=1$.}
 \begin{prop}
 \label{prop:exam-1}
  The optimal DMT of the SISO streaming setup with two messages per coherence block, $\Delta=0$, and $T=1$, is: 
 \begin{align}
 d(r) = \min\left(1-\frac{r}{2}, 2-2r\right), \quad r \in [0,1].\label{eq:DMT_ex}
 \end{align}
 \end{prop}

\subsubsection*{Converse}

The upper bound is based on two genie aided arguments. The bound $d(r) = 1-r/2$  follows by revealing every message $\rvw_{k,2}$ to the destination. Thus message $\rvw_{k,1}$ needs to be decoded at the end of coherence block $k$. Since $\rvw_{k,1}$ is uniformly distributed in $\{1,2,\ldots, 2^{MR/2}\}$ and has a rate of $R/2$, it follows that the DMT for this genie aided channel equals $d(r) = 1-r/2$.

To establish the other upper bound of  $d(r) = 2-2r$ we consider another genie aided channel.  We reveal message $\rvw_{k,2}$ at the start of  coherence block $k$ and relax the deadline of $\rvw_{k,1}$ and $\rvw_{k,2}$ such that both only need to be decoded at the end of the coherence block ${k+1}$. Such an assumption can clearly only improve the DMT. However the setup now is identical to that considered in  Theorem~\ref{thm:streaming_DMT} where the message $\rvw_k = (\rvw_{k,1},\rvw_{k,2})$ arrive at the start of coherence block $k$ and must be decoded with a delay of $T=2$ coherence blocks.  The associated DMT,  $d(r) =2-2r$ for this channel, is thus an upper bound for the original setup. This completes the justification of the converse.

\subsubsection*{Achievability}
We next present a coding scheme that attains the DMT in~\eqref{eq:DMT_ex}.
We first split each message $\rvw_{k,1}$ into two equal sized messages $\rvw_{k,1} = (\rvw_{k,1}^{1}, \rvw_{k,1}^{2})$, where each sub-message is of rate $R_0 = R/4$.  Thus we can assume that both $\rvw_{k,1}^1$ and $\rvw_{k,1}^2$ are independent and sampled uniformly from $\cJ_M = \{1,2,\ldots, 2^{MR/4}\}$.  We do not split the messages $\rvw_{k,2}$ and assume that it is sampled uniformly from $\cI_M = \{1,2,\ldots, 2^{MR/2}\}$. We sample three Gaussian codebooks as follows:
\begin{itemize}
\item The codebook $\cC_{A}$ consisting of $2^{3MR_0}$ codewords $\rvx_A^{M/2}$ sampled i.i.d. from $\CN(0,\rho)$. Each pair $(\rvw_{k,1}^1, \rvw_{k-1,2})$ is mapped to a unique codeword $\rvbx_A(\rvw_{k,1}^1, \rvw_{k-1,2})$ i.e.,
\begin{align}
\cC_A = \left\{ \rvbx_A(\rvw_{k,1}^1, \rvw_{k-1,2})\right\}_{\rvw_{k,1}^1 \in \cJ_M, \rvw_{k-1,2} \in \cI_M}
\end{align}
\item The codebook $\cC_{B}$ consisting of $2^{MR_0}$ codewords $\rvx_{B}^{M/2}$ sampled i.i.d. from $\CN(0, \rho)$. Each message $\rvw_{k,1}^2$ is mapped to a unique codeword $\rvbx_{B}(\rvw_{k,1}^2)$ i.e.,
\begin{align}
\cC_B = \left\{\rvbx_B(\rvw_{k,1}^2)\right\}_{\rvw_{k,1}^2 \in \cJ_M}
\end{align}

\item The codebook $\cC_{C}$ consisting of $2^{2MR_0}$ codewords $\rvx_C^{M/2}$ sampled i.i.d. from $\CN(0, \rho^{1-\beta})$. Each message $\rvw_{k,2}$ is mapped to a unique codeword $\rvbx_{C}(\rvw_{k,2})$. 
\begin{align}
\cC_C = \left\{\rvbx_C(\rvw_{k,2})\right\}_{\rvw_{k,2} \in \cI_M}
\end{align}

We will select $\beta=r/2$. Note that the total power  in the second block is $\rho +\rho^{1-r/2} \doteq \rho$, since $\rho^{1-r/2} \ll \rho$. 
\begin{figure*}
\begin{center}
\includegraphics[width=\linewidth]{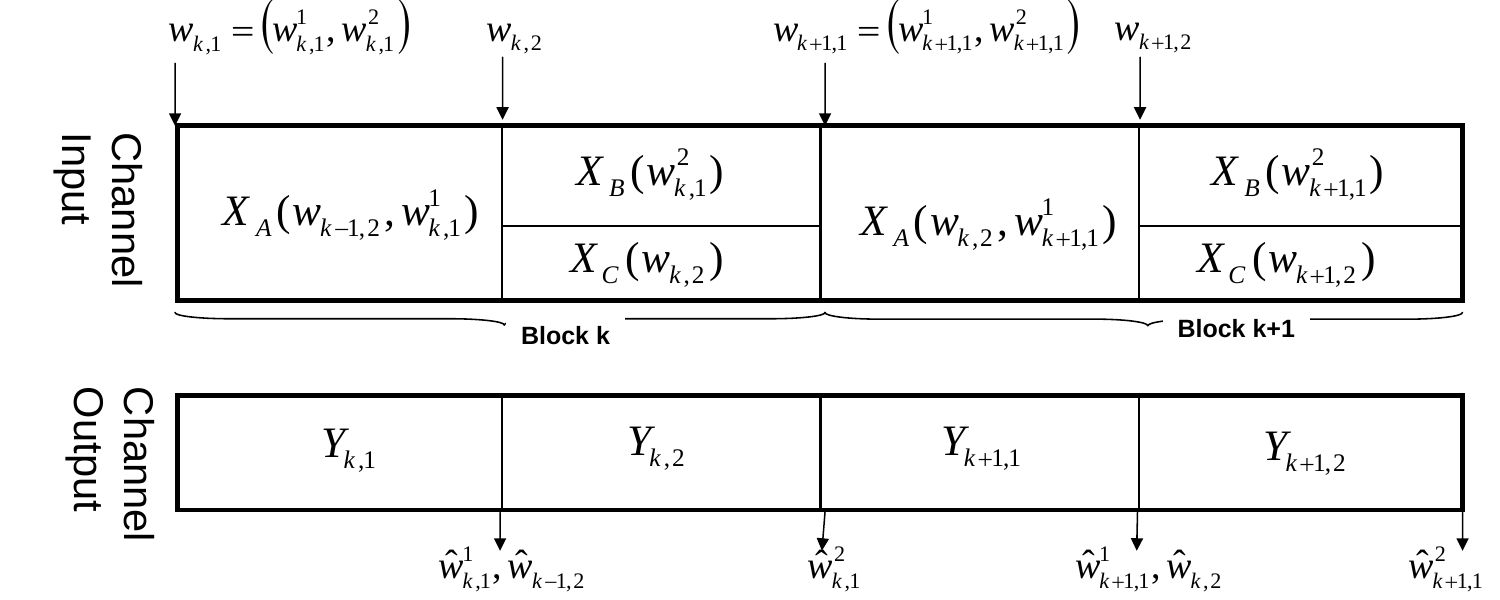}
\caption{Coding scheme for two messages per coherence block with $\Delta=0$. The first message $\rvw_{k,1}$ is split into two sub-messages $(\rvw_{k,1}^1, \rvw_{k,1}^2)$ of equal size, while the second message $\rvw_{k,2}$ is not split. In the first half of the coherence block, we transmit the codeword $\rvbx_A(\rvw_{k-1,2},\rvw_{k,1}^1)$ while in the second half of the coherence block we transmit the sum $\rvbx_B(\rvw_{k,1}^2) + \rvbx_C(\rvw_{k,2})$.} 
\end{center}
\end{figure*}

\end{itemize}

In coherence block $k,$ the transmitter transmits $\rvbx_{k,1} = \rvbx_{A}(\rvw_{k,1}^1, \rvw_{k-1,2})$ in the first half of the coherence block and $\rvbx_{k,2} = \rvbx_B(\rvw_{k,1}^2) + \rvbx_C(\rvw_{k,2})$ in the second half of the coherence block. The receiver observes $\rvby_{k,i} = \rvh_k \cdot~\rvbx_{k,i} + \rvbz_{k,i}$
for $i \in \{1, 2\}$. The decoding of the messages is as follows. 

\begin{itemize}
\item {In the second half of coherence block $k$, the receiver decodes $\rvw_{k,1}^2$ using $\rvby_{k,2}$, by treating $\rvbx_C(\rvw_{k,2})$ as additional noise. It searches for a unique message $\hat{\rvw}_{k,1}^2 \in \cJ_M$ such that}
$\left(\rvbx_B(\hat{\rvw}_{k,1}^2), \rvby_{k,2}\right) \in \cT_{\eps,2}^{M/2}.$ The error event $\cE_{k,1}^2$ denotes the event that $\hat{\rvw}_{k,1}^2 \neq {\rvw}_{k,1}^2$ and let $\cO_{k,1}^2$ denote the outage event:
\begin{align}
 \left\{\frac{1}{2}\log\left(1 + \frac{|\rvh_k|^2 \rho}{1 + |\rvh_k|^2 \rho^{1-\beta}}\right) \le \frac{r}{4} \log \rho\right\}.\label{eq:outage2}
\end{align}

\item After decoding $\hat{\rvw}_{k,1}^2$ the decoder subtracts $\rvbx_{B}(\hat{\rvw}_{k,1}^2)$ from $\rvby_{k, 2}$ i.e., $\tilde{\rvby}_{k,2}= \rvby_{k, 2} - \rvh_k \rvbx_B(\hat{\rvw}_{k,1}^2)$. The decoder uses the second half of coherence block $k$ and the first half of coherence block ${k+1}$ to decode the message pair $(\rvw_{k,2},\rvw_{k+1,1}^{1} )$. 
In particular, it searches for a pair $(\hat{\rvw}_{k,2},\hat{\rvw}_{k+1,1}^1)$ such that $(\rvbx_C(\hat{\rvw}_{k,2}), \tilde{\rvby}_{k,2}) \in \cT_{\eps,3}^{M/2}$ and $(\rvbx_A(\hat{\rvw}_{k+1,1}^1, \hat{\rvw}_{k,2}), \rvby_{k+1,1}) \in \cT_{\eps,4}^{M/2}$ are jointly typical.  

The error analysis involves two events,  $\cE_{k,2}$ and $\cE_{k+1,1}^{1}$, associated with the error in decoding $\rvw_{k,2}$ and $\rvw_{k+1,1}^1$ respectively.
In particular, let  $\cE_{k,2}= \{\hat{\rvw}_{k,2} \neq \rvw_{k,2}\}$  and the outage event $\cO_{k,2}$ be given by:
\begin{equation}
\bigg\{\frac{1}{2}\log\left(1 + \rho^{1-\beta} |\rvh_k|^2\right) + \frac{1}{2}\log(1+ \rho|\rvh_{k+1}|^2) \le \frac{3}{4}r\log\rho\bigg\}.
\label{eq:outage-k2}
\end{equation}
Similarly let $\cE_{k+1,1}^1 = \{\{\hat{\rvw}_{k,2} = \rvw_{k,2} \} \cap \{\hat{\rvw}_{k+1,1}^1 \neq \rvw_{k+1,1}^1\}\}$  be the event that the message $\rvw_{k,2}$
is decoded correctly, but an error occurs in the decoding of $\rvw_{k+1,1}^1$ and  let $\cO_{k,1}^1$ be the corresponding outage event: \begin{align}\left\{\frac{1}{2}\log(1+\rho |\rvh_k|^2) \le \frac{r}{4}\log\rho\right\}.\label{eq:outage1}\end{align}
\end{itemize}

It suffices to show that the error probability satisfies $\Pr\left(\cE_{k,1}^2 \cup \cE_{k,2} \cup \cE_{k+1,1}^1\right) \stackrel{\cdot}{\le} \rho^{-d(r)} +o_M(1)$ where $d(r)$ is defined in~\eqref{eq:DMT_ex} and $o_M(1)$ approaches zero as $M \rightarrow \infty$. By selecting $M$ sufficiently large (for each fixed $\rho$), the proposed DMT is then achievable.

We first consider the event $\cE_{k,1}^2 = \{\hat{\rvw}_{k,1}^2 \neq \rvw_{k,1}^2\}$ and use the following upper bound:
{\allowdisplaybreaks{\begin{align}
\Pr\left(\cE_{k,1}^2 \right) &\le \Pr\left(\cE_{k,1}^2~|~\cO_{k,1}^{2,c}\right) + \Pr(\cO_{k,1}^2) \label{eq:ek-21}
\end{align}}}
where $\cO_{k,1}^2$ is defined in~\eqref{eq:outage2}, with $\beta =r/2$ as:
\begin{align}
 \left\{\frac{1}{2}\log\left(1 + \frac{|\rvh_k|^2 \rho}{1 + |\rvh_k|^2 \rho^{1-r/2}}\right) \le \frac{r}{4} \log \rho\right\}.
\label{eq:outage22}\end{align} From standard arguments, the first term in~\eqref{eq:ek-21} decreases to zero as $M \rightarrow\infty$, and thus we only need to upper bound the second term. Letting $|\rvh_k|^2 = \rho^{-(1-\al)}$ we have that~\eqref{eq:outage22} is equivalent to
\begin{align}
\log\left(1 + \frac{\rho^{\al}}{1+\rho^{\al-r/2}}\right) \le \frac{r}{2}\log\rho
\end{align}
which in turn implies that $\alpha < \frac{r}{2}$ as $\rho \rightarrow \infty$. Thus we have that
\begin{align}
\Pr(\cO_{k,1}^2) \stackrel{\cdot}{\le} \rho^{-(1-r/2)}
\end{align}
and in turn
\begin{align}
\Pr(\cE_{k,1}^2) \stackrel{\cdot}{\le} \rho^{-(1-r/2)} + o_M(1) \label{eq:ek-22}
\end{align}
holds. 

Next we upper bound the probability of the event $\cE_{k,2} = \{\rvw_{k,2} \neq \hat{\rvw}_{k,2}\}$. We can express 
{\allowdisplaybreaks{\begin{align}
\Pr(\cE_{k,2}) &=\Pr(\cE_{k,2}|\cO_{k,2}^c) + \Pr(\cO_{k,2})
\end{align}}}
where recall that the event $\cO_{k,2}$ in~\eqref{eq:outage-k2} is defined, with $\beta=r/2$ as:
\begin{align}
\bigg\{\frac{1}{2}\log\left(1 + \rho^{1-r/2} |\rvh_k|^2\right) + \frac{1}{2}\log(1+ \rho|\rvh_{k+1}|^2) \le \frac{3}{4}r\log\rho\bigg\}. \label{eq:outage-k22}
\end{align} Note that whenever $\{\rvw_{k,2} \neq \hat{\rvw}_{k,2}\}$, we have that  
$(\rvbx_C(\hat{\rvw}_{k,2}), \tilde{\rvby}_{k,2})$ are mutually independent and furthermore
$(\rvbx_A(\hat{\rvw}_{k+1,1}^1, \hat{\rvw}_{k,2}), \rvby_{k+1,1})$ are mutually independent. 
It can be shown through a standard union bound argument that $\Pr(\cE_{k,2}|\cO_{k,2}^c)$ vanishes to zero as $M\rightarrow\infty$. 
To upper bound $\cO_{k,2}$ we let $|\rvh_k|^2 = \rho^{-(1-\al_1)}$ and $|\rvh_{k+1}|^2 = \rho^{-(1-\al_2)}$ and note that~\eqref{eq:outage-k22} reduces to:
\begin{align}
\log(1+ \rho^{\al_1-r/2}) + \log(1+\rho^{\al_2}) \le \frac{3r}{2}\log\rho.
\end{align}
The associated DMT is given by
\begin{align}
d_2(r) &= \min_{(\al_1, \al_2)\in \cA} (1-\al_1)^+ + (1-\al_2)^+
\end{align}where $\cA = \{(\al_1, \al_2)\ge 0 :  (\al_1-r/2)^+ + \al_2 \le 3r/2\}$ and $(v)^+$ equals $0$ is $v <0$.
It can be deduced that $d_2(r) = 2-2r$ and thus
\begin{align}
\Pr(\cO_{k,2})\stackrel{\cdot}{\le} \rho^{-2(1-r)}
\end{align}holds. Thus we have that
\begin{align}
\Pr(\cE_{k,2}) &\stackrel{\cdot}{\le} o_M(1) + \rho^{-2(1-r)} \label{eq:ek-23}
\end{align}holds.

Finally we consider the event $\cE_{k+1,1}^1 = \{\{\hat{\rvw}_{k,2} =
\rvw_{k,2}\} \cap \{\hat{\rvw}_{k+1,1}^1 \neq \rvw_{k+1,1}^1\}\}$
which corresponds to an error in message estimate
$\hat{\rvw}_{k+1,1}^1$ in the first half of the coherence block. Under
this event the codeword $\rvbx_C(\hat{\rvw}_{k,2})$ is decoded
correctly however the pair $(\rvbx_A(\hat{\rvw}_{k+1,1}^1,
\hat{\rvw}_{k,2}), \rvby_{k+1,1})$ is mutually independent. Using
$\cO_{k,1}^1$ be defined in~\eqref{eq:outage1}, we can express
\begin{align}
\Pr\left(\cE_{k,1}^1\right) &\le \Pr(\cO_{k,1}^1) +
\Pr\left(\cE_{k,1}^1~\big|~\cO_{k,1}^{1,c}\right).
\end{align}
It follows from standard arguments that $\Pr(\cO_{k,1}^1) \doteq \rho^{-(1-r/2)}$ and furthermore
$\Pr\left(\cE_{k,1}^1~\big|~\cO_{k,1}^{1,c}\right)$ vanishes to zero as $M \rightarrow \infty$.
Thus we have that
\begin{align}
\Pr\left(\cE_{k,1}^1\right) &\stackrel{\cdot}{\le} \rho^{-(1-r/2)} + o_M(1). \label{eq:e1-bnd}
\end{align}

This completes our achievability.

\color{black}
\iffalse
we have that
\begin{align}
\Pr(\cE_k) \stackrel{\cdot}{\le} o_M(1) + \rho^{\min(2-2r, 1-r/2)}.
\end{align}
\fi

\subsection{Unknown Offset}
\label{subsec:unknown}
We consider the case when either $\Delta = 0$ or $\Delta = 1/2$, but
when the actual value of $\Delta$ is not known to the
transmitter. Such a setup applies when simultaneously transmitting to
two users whose coherence blocks are offset by $M/2$ symbols. The
following result shows that we cannot have a universal coding scheme
oblivious of $\Delta$ that achieves the same DMT.
\begin{prop}
 \label{prop:exam-2}
Consider the SISO channel model with two messages in each coherence
block as in Prop.~\ref{prop:exam-1}. Assume that either $\Delta =0$ or
$\Delta= 1/2$, but where the actual value of $\Delta$ is known only to
the receiver.  The DMT for this setup equals $d(r) = 1-r$.
\end{prop}
 \begin{new-proof}

 The achievability is straightforward. Each message $\rvw_{k,j}$ is
 mapped to a codeword of length $M/2$ of a Gaussian codebook and
 transmitted immediately.  Since each message is of rate
 $\frac{r}{2}\log\rho$ a DMT of $d(r) = 1-r$ is achievable.
 
 For the converse, we consider a multicast setup with two
 receivers. In coherence block $k$ the transmitter transmits
 $\rvbx_{k,1}$ in the first half of the coherence block and transmits
 $\rvbx_{k,2}$ in the second half, i.e., $\rvbx_k =
 [\rvbx_{k,1}~\rvbx_{k,2}],$ where both $\rvbx_{k,1}, \rvbx_{k,2}\in
 {\mathbb C}^{M/2}$.  Receiver $1$ observes $\rvby_k =
 [\rvby_{k,1}~\rvby_{k,2}]$ in coherence block $k$ as follows:
 {\allowdisplaybreaks{\begin{align}
 \rvby_{k, 1} &= \rvh_k \rvbx_{k,1} + \rvbn_{k,1}, \\
 \rvby_{k,2} &= \rvh_k \rvbx_{k,2} + \rvbn_{k,2}.
 \end{align}
 Receiver $2$ observes $\rvbv_k = [\rvbv_{k,1}~\rvbv_{k,2}]$ in coherence block $k$ as follows:
 \begin{align}
 \rvbv_{k,1} &= \rvh_{k}\rvbx_{k,1} + \rvbz_{k,1}, \\
 \rvbv_{k,2} &=\rvh_{k+1}\rvbx_{k,2} + \rvbz_{k,2}.
 \end{align}}}
 The noise variables $\rvbn_{k,j}$ and $\rvbz_{k,j}$ have
 i.i.d.\ $\CN(0,1)$ entries.  For both receivers we have $T=1$ and
 message deadlines are shown in Fig.~\ref{fig:multiple}.  Note that
 the duration of $\rvw_{k,1}$ spans only one coherence block for
 receiver $1$, but $\rvw_{k,2}$ spans two coherence blocks.  Likewise
 $\rvw_{k,2}$ spans only one coherence block for receiver $2$, but
 $\rvw_{k,1}$ spans two blocks.  We show that under this constraint
 the maximum possible DMT is $d(r)=1-r$
 
 We begin by considering Fano's inequality for receiver $1$ for message $\rvw_{0,1}$ and rate $\frac{Mr}{2}\log\rho$:
 \begin{align}
 \Pr(\cE_{0,1}|\rvh_0=h_0) \ge 1- \frac{2}{Mr\log\rho} -\frac{2I(\rvw_{0,1}; \rvby_1|\rvh_0=h_0)}{Mr\log\rho}.
 \end{align}
 Ignoring the second term, which goes to zero as $M\rightarrow\infty$,
 and using the same sequence of steps leading to~\eqref{eq.condFanos} we have with $P_\del \doteq \rho^{-(1-r+\del)}$
{\allowdisplaybreaks{ \begin{align}
 \Pr(\cE_{0,1}) &\ge P_\delta\left(1-\frac{2I(\rvw_{0,1};\rvby_0 |\rvh_0,\rvh_0 \in \cH_\del)}{Mr\log\rho}\right)\\
 &= P_\delta\left(1-\frac{2I(\rvw_{0,1};\rvby_0 |\rvh_0^{N+1},\rvh_0^{N+1} \in \cH_\del^{N+2})}{Mr\log\rho}\right)\label{eq:chan-indep2}\\
 &\ge P_\del \left(1-\frac{2I(\rvw_{0,1};\rvby_0^{N}, \rvbv_0^{N} |\rvh_0^{N+1},\rvh_0^{N+1} \in \cH_\del^{N+2})}{Mr\log\rho}\right)\label{eq:w0term}
 \end{align}}}
 where~\eqref{eq:chan-indep2} follows from the fact that ${\rvh_1^{N+1}}$ is independent of $(\rvw_0, \rvby_0,\rvh_0)$.

 Similarly, applying Fano's inequality to receiver $2$ for message $\rvw_{0,2}$ we have
 \begin{align}
&\Pr(\cE_{0,2}) \ge P_\del\left(1-\frac{2I(\rvw_{0,2};\rvbv_{0,2},\rvbv_{1,1}|\rvw_{0,1},\rvh_1, \rvh_1 \in \cH_\del)}{Mr\log\rho}\right)\\
 &=P_\del\left(1-\frac{2I(\rvw_{0,2};\rvbv_{0,2},\rvbv_{1,1}|\rvw_{0,1}, \rvh_0^{N+1},\rvh_0^{N+1} \in \cH_\del^{N+2})}{Mr\log\rho}\right)\\
 &\ge P_\del\left(1-\frac{2I(\rvw_{0,2};\rvby_0^{N},\rvbv_0^{N}|\rvw_{0,1}, \rvh_0^{N+1},\rvh_0^{N+1} \in \cH_\del^{N+2})}{Mr\log\rho}\right).
 \end{align}

Likewise we can show that for each $k \le N-1$
\begin{align}
&\Pr(\cE_{k,1})\ge\notag\\ & P_\del\left(1- \frac{2 I(\rvw_{k,1};\rvby_0^{N}, \rvbv_0^{N} |\rvh_0^{N+1},\rvh_0^{N+1}\in \cH_\del^{N+2}, \rvw_0^{k-1})}{Mr\log\rho}\right)\\
&\Pr(\cE_{k,2}) \ge \notag \\ &P_\del\left(1-\frac{2I(\rvw_{k,2};\rvby_0^{N}, \rvbv_0^{N}|\rvh_0^{N+1},\rvh_0^{N+1} \in \cH_\del^{N+2},\rvw_0^{k-1},\rvw_{k,1})}{Mr\log\rho}\right).
\end{align}

Thus we have that
{\allowdisplaybreaks{\begin{align}
&\max_{0\le k \le N-1} \max\left\{\Pr(\cE_{k,1}),\Pr(\cE_{k,2}))\right\} \notag\\ &\ge \frac{1}{2N} \sum_{k=0}^{N-1}\left\{\Pr(\cE_{k,1}) + \Pr(\cE_{k,2})\right\}\\
&\ge P_\del\left(1 - \frac{I(\rvw_0^{N-1};\rvby_0^N, \rvbv_0^N|\rvh_0^{N+1}\in \cH_\del^{N+1})}{NMr\log\rho}\right)\\
&\ge P_\del\left(1 - \frac{I(\rvbx_0^{N};\rvby_0^N, \rvbv_0^N|\rvh_0^{N+1}\in \cH_\del^{N+1})}{NMr\log\rho}\right)\\
&\ge P_\del \times \notag\\ &\bigg(\!\!1 \!\!- \!\!\frac{\sum_{k=0}^N I(\rvbx_{k,1};\rvby_{k,1}, \rvbv_{k,1}|\rvh_k) + I(\rvbx_{k,2};\rvby_{k,2},\rvbv_{k,2}|\rvh_k, \rvh_{k+1})}{NMr\log\rho}\!\!\bigg)\\
&\ge \rho^{-(1-r+\del)} \left(1- \frac{N+1 + (N+1)(r-\del)\log\rho}{Nr\log\rho}\right).\label{eq:multicast-lb}
\end{align}}}
The steps leading to~\eqref{eq:multicast-lb} are similar
to~\eqref{eq:ChGain2} and hence are not elaborated.  For $N$
sufficiently large the expression in the brackets
in~\eqref{eq:multicast-lb} is positive. This establishes that $d(r)
\ge 1-r+\del$ must hold. Since $\del>0$ is arbitrary this concludes
the converse in Prop.~\ref{prop:exam-2}. \end{new-proof}
 
We conclude this section with the following remark.  When there are
multiple messages that arrive at equal intervals in each coherence
block, different messages observe different channel
conditions. Prop.~\ref{prop:exam-1} shows that coding schemes that
exploit this asymmetry between messages can improve the DMT.  On the
other hand such schemes depend crucially on where the messages arrive
in each block. If such information is not available the DMT is, in
general, smaller, as is established in Prop.~\ref{prop:exam-2}.

\section{Conclusions}
\label{sec.conclusion}

In this paper we study the problem of delay-constrained streaming over
a block fading channel.  We establish the diversity multiplexing
tradeoff when there is one message arriving in each coherence
block. The converse is based on a novel ``outage-amplification''
argument that builds up a contradiction, over a sufficiently large
duration, if we assume a larger DMT.  We propose two coding schemes
for achieving the optimal DMT. The first uses an interleaving scheme
that reduces the system to a set of parallel independent channels. The
advantage of this scheme is its simplicity.  The disadvantage is that
the playback deadline $T$ must be known in advance.  The second scheme
pairs a sequential decoder with a tree code.  This scheme also attains
that DMT, and in a delay-universal fashion, but is more computationally
complex and appears to require sufficiently long coherence blocks.
Finally, we discuss some extensions when multiple messages arrive in
each coherence block.
 
The fundamental limits of delay-constrained streaming over fading
channels are not well understood in general.  We hope that the
techniques developed in this work can serve as a useful starting point
for other investigations.

\appendices
\section{Proof of the Lower Bound in~\eqref{eq.condFanos}}
\label{app:outage}

In this Appendix we incorporate the effect of outage into our lower bound.
We continue from~(\ref{eq.pointwise}), dividing the channel
realizations into sets in which the suffix is in outage, and when it is not, and dropping the latter.  Thus,
\begin{equation*}
\Pr[\cE_k]  \geq \sum_{\bH_0^\Tk :  \bH_k^\Tk \in  \cH_\del^T} \Pr[\rvbH_0^{\Tk} = \bH_0^{\Tk}] \Pr[\cE_k| \rvbH_0^{\Tk} = \bH_0^{\Tk}].
\end{equation*}
Applying~(\ref{eq.fanos}) and marginalizing out over the prefixes
$\{\bH_0^{k-1}\}$ we get
\begin{multline}
\Pr[\cE_k]  \geq   \sum_{\bH_k^\Tk :  \bH_k^\Tk \in  \cH_\del^T}  \Pr[\rvbH_k^{\Tk} = \bH_k^{\Tk}] \times \\ \left(1 -
  \frac{1}{Mr\log\rho} - \frac{I(\rvw_k;  \rvbY_k^{\Tk}| \rvw_0^{k-1}, \rvbH_k^{\Tk}= \bH_k^{\Tk})}{Mr\log\rho}\right).
\label{eq:Ek-lb}
\end{multline}

We can further express the right hand side of~\eqref{eq:Ek-lb} as follows. Note that
\begin{align}
&\sum_{\bH_k^\Tk :  \bH_k^\Tk \in  \cH_\del^T}  \Pr[\rvbH_k^{\Tk} = \bH_k^{\Tk}] \left(1 -
  \frac{1}{Mr\log\rho}\right) \\&= \left(1 -
  \frac{1}{Mr\log\rho}\right)\sum_{\bH_k^\Tk :  \bH_k^\Tk \in  \cH_\del^T}  \Pr[\rvbH_k^{\Tk}]\\
  &=  \left(1 -
  \frac{1}{Mr\log\rho}\right) \Pr[\rvbH_k^\Tk \in  \cH_\del^T] \label{eq:Ek-lb-1}
\end{align}

For the second term in~\eqref{eq:Ek-lb}, note that
\allowdisplaybreaks{\begin{align}
&\sum_{\bH_k^\Tk :  \bH_k^\Tk \in  \cH_\del^T} \Pr[\rvbH_k^{\Tk} = \bH_k^{\Tk}] I(\rvw_k;  \rvbY_k^{\Tk}| \rvw_0^{k-1}, \rvbH_k^{\Tk}= \bH_k^{\Tk}) \\
&=\sum_{\bH_k^\Tk :  \bH_k^\Tk \in  \cH_\del^T} \Pr[\rvbH_k^{\Tk} = \bH_k^{\Tk}] \times \notag \\ &\qquad\qquad I(\rvw_k;  \rvbY_k^{\Tk}| \rvw_0^{k-1}, \rvbH_k^{\Tk}= \bH_k^{\Tk}, \rvbH_k^\Tk \in \cH_\del^T) \label{eq:mut-inf-cond}\\
&=\Pr[\rvbH_k^\Tk \in \cH_\del^T]\sum_{\bH_k^\Tk :  \bH_k^\Tk \in  \cH_\del^T} \Pr[\rvbH_k^{\Tk} = \bH_k^{\Tk} | \rvbH_k^\Tk \in \cH_\del^T]\times\notag\\ & \qquad\qquad I(\rvw_k;  \rvbY_k^{\Tk}| \rvw_0^{k-1}, \rvbH_k^{\Tk}= \bH_k^{\Tk}, \rvbH_k^\Tk \in \cH_\del^T) \label{eq:bayes-app}\\
&=\Pr[\rvbH_k^\Tk \in \cH_\del^T] I(\rvw_k;  \rvbY_k^{\Tk}| \rvw_0^{k-1}, \rvbH_k^{\Tk}, \rvbH_k^\Tk \in \cH_\del^T), \label{eq:avg-mut-inf}
\end{align}}
where in~\eqref{eq:mut-inf-cond}, we use the fact that the indicator random variable ${\rvbH_k^\Tk \in \cH_\del^T}$ is a deterministic function of $\rvbH_k^\Tk$, and hence can be added as in~\eqref{eq:mut-inf-cond}. In~\eqref{eq:bayes-app}, we note that for each $\bH_k^\Tk \in  \cH_\del^T$ we have:
\begin{align}
&\Pr[\rvbH_k^{\Tk} = \bH_k^{\Tk} | \rvbH_k^\Tk \in \cH_\del^T] \\&= \frac{\Pr[\rvbH_k^{\Tk} = \bH_k^{\Tk}] \Pr[\rvbH_k^\Tk \in \cH_\del^T|\rvbH_k^{\Tk} = \bH_k^{\Tk}]}{\Pr[\rvbH_k^\Tk \in \cH_\del^T]} \\
&=\frac{\Pr[\rvbH_k^{\Tk} = \bH_k^{\Tk}]}{\Pr[\rvbH_k^\Tk \in \cH_\del^T]}.\label{eq:bayes}
\end{align}
Thus substituting~\eqref{eq:bayes} in~\eqref{eq:mut-inf-cond} we have that~\eqref{eq:bayes-app} follows.

Substituting~\eqref{eq:Ek-lb-1} and~\eqref{eq:avg-mut-inf} into~\eqref{eq:Ek-lb}, we have
\begin{align}
\Pr[\cE_k] & \ge \Pr[\rvbH_k^\Tk \in \cH_\del^T]\times \notag\\ &\left(1 -
  \frac{1}{Mr\log\rho} - \frac{I(\rvw_k;  \rvbY_k^{\Tk}| \rvw_0^{k-1}, \rvbH_k^{\Tk}, \rvbH_k^\Tk \in \cH_\del^T)}{Mr\log\rho} \right).\label{eq.condFanos2}
%& = \Pr[\rvbH_k^{\Tk} \in \cH_\del^T] \Pr[\cE_k | \rvbH_k^{\Tk} \in \cH_\del^T].
\end{align}
This completes the justification of~\eqref{eq.condFanos}.

\section{Proof of~\eqref{eq:PrE_A_Bound}.}
\label{app:PrE_A_proof}

Our proof is based on the Chernoff-Cramer theorem of large deviations stated below.
\begin{thm}
Suppose that $\rvx_1,\ldots, \rvx_N$ are i.i.d.\ random variables with a rate function $f_\rvx(\cdot)$
defined as
\begin{align}
f_\rvx(t) = \sup_{\theta}\bigg\{\theta \cdot t - \log E_\rvx\left[\exp(\theta \cdot \rvx)\right]\bigg\},
\label{eq:rate-f}
\end{align}
and let $M_n = \frac{1}{n}\sum_{i=1}^n \rvx_i$. Then there exists a constant $N>0$ such that for all $n\ge N$ 
\begin{align}
\Pr(M_n \ge t) \le e^{-n f_\rvx(t)}. \label{cramer-bnd}
\end{align}
\hfill$\Box$
\label{thm:chernoff}
\end{thm}

Recall that $\cA_{l,l'}$ is the event that the true codeword is not
jointly typical with the received sequence.  To upper bound the probability 
we can ignore the marginal
typicality constraints and use
\begin{align}
&\Pr(\cA_{l,l'}) \leq \notag\\ &\Pr\left( \bigg|
  \frac{\sum_{k=l}^{l'}[- \log p_{\rvbX_k, \rvbY_k}(\rvbX_k, \rvbY_k)-
    h(p_{\rvbX_k, \rvbY_k})]}{M(l' - l + 1)}\bigg| >
  \eps\right).
\label{eq:MsgSeq_Atypical_2}
\end{align}
Note that as
$\rvbY_k = \rvbH_k \rvbX_k + \rvbZ_k$, the $\rvbH_k$ are known to the
decoder, and the noise sequence $\{\rvbZ_k\}$ is independent,
\begin{align}
p_{\rvbX_k, \rvbY_k}(\svbX_k, \svbY_k) &= p_{\rvbX_k}(\svbX_k)
p_{\rvbY_k| \rvbX_k}(\svbY_k| \svbX_k)\\ &= p_{\rvbX_k}(\svbX_k)
p_{\rvbZ_k}(\svbY_k - \bH_k \cdot \svbX_k)\\ &= p_{\rvbX}(\svbX_k)
p_{\rvbZ_k}(\svbZ),
\end{align} 
where the last equality holds since the codewords are sampled
i.i.d.\ and the noise is also i.i.d.  Thus
$h(p_{\rvbX_k,\rvbY_k})= h(p_\rvbX) + h(p_\rvbZ)$.  And so
\begin{align}
& \bigg|\sum_{k=l}^{l'}[ -\log p_{\rvbX_k, \rvbY_k}(\rvbX_k, \rvbY_k)-
    h(p_{\rvbX_k, \rvbY_k})]\bigg|\\ & = 
\bigg|\sum_{k=l}^{l'}[-\log p_{\rvbX}(\rvbX_k)-\log
  p_{\rvbZ}(\rvbZ_k) - h(p_{\rvbX}) -
  h(p_{\rvbZ})]\bigg|\notag\\ &\le \bigg|\sum_{k=l}^{l'}[-\log
  p_{\rvbX}(\rvbX_k) - h(p_{\rvbX})\bigg| + 
  \bigg|\sum_{k=l}^{l'} [-\log p_{\rvbZ}(\rvbZ_k) - 
  h(p_{\rvbZ})]\bigg|\label{eq:triangular_ineq}
\end{align}
where the last step follows from the triangle
inequality. Substituting~\eqref{eq:triangular_ineq}
into~\eqref{eq:MsgSeq_Atypical_2} and using using the union bound we
have
\begin{align}
\Pr(\cA_{l,l'}) \le \Pr(\cA_{l,l'}^{X})+ \Pr(\cA_{l,l'}^{Z})
\end{align}
where we define
\begin{align}
\cA_{l,l'}^X = \left\{ \svbX_l^{l'} : \bigg|\frac{\sum_{k=l}^{l'}[-\log
  p_{\rvbX}(\svbX_k) - h(p_{\rvbX})]}{M(l'-l+1)}\bigg| \ge
\eps\right\},\\ \cA_{l,l'}^Z= \left\{\svbZ^l_{l'} :
\bigg|\frac{\sum_{k=l}^{l'}[-\log p_{\rvbZ}(\svbZ_k) -
  h(p_{\rvbZ})}{M(l'-l+1)}\bigg|\ge \eps\right\}.
\end{align} 
Note that $\rvbX_k $ is a sequence of $M$ i.i.d.\ random vectors each
sampled from $\CN(0, \frac{\rho}{\Nt}\bI)$ and $E[-\log
  p_{\rvbX}(\rvbX_k)] = h(p_{\rvbX})$.  Similarly, $E[-\log
  p_{\rvbZ}(\rvbZ_k)] = h(p_{\rvbZ})$.  Then using Theorem~\ref{thm:chernoff}, there exist functions $f_X(\eps)$ and
$f_Z(\eps)$  such that for sufficiently large  $N = M(l'-l+1)$
\begin{align*}
\Pr(\cA_{l,l'}^X) \le \exp\{-M(l'-l+1)f_X(\eps)\},\\
\Pr(\cA_{l,l'}^Z) \le \exp\{-M(l'-l+1)f_Z(\eps)\}.\\
\end{align*}
Furthermore by directly using~\eqref{eq:rate-f} we can show that $f_X(\eps)>0$ and $f_Y(\eps)>0$.
Setting $f(\eps) = \max(f_X(\eps),f_Z(\eps))$
establishes~\eqref{eq:PrE_A_Bound}.

%\bibliographystyle{IEEEtran}
%\bibliography{sm}

\vspace{-4em}
%\begin{IEEEbiography}[{\includegraphics[width=0.75in,height=1.25in,clip,keepaspectratio]{./Ashish_Khisti}}]{Ashish Khisti}
\begin{IEEEbiographynophoto}{Ashish Khisti}
Ashish Khisti is an assistant professor in the Electrical and Computer
Engineering (ECE) department at the University of Toronto, Toronto,
Ontario Canada. He received his BASc degree in Engineering Sciences
from University of Toronto and his S.M and Ph.D. Degrees from the
Massachusetts Institute of Technology (MIT), Cambridge, MA, USA. His
research interests span the areas of information theory, wireless
physical layer security and streaming in multimedia communication
systems. At the University of Toronto, he heads the signals,
multimedia and security laboratory. For his graduate studies he was a
recipient of the NSERC postgraduate fellowship, HP/MIT alliance
fellowship, Harold H. Hazen Teaching award and the Morris Joseph Levin
Masterworks award.
\end{IEEEbiographynophoto}
%\end{IEEEbiography}

\begin{IEEEbiographynophoto}{Stark C.~Draper}(S'99-M'03) received the
M.S. and Ph.D. degrees in electrical engineering and computer science
from the Massachusetts Institute of Technology (MIT), and the B.S. and
B.A. degrees in electrical engineering and history, respectively, from
Stanford University.

He is an Associate Professor of Electrical and Computer Engineering at
the University of Toronto, Canada.  From 2007-2014 he was an Associate
Professor at the University of Wisconsin, Madison. Before moving to
the University of Wisconsin he was with the Mitsubishi Electric
Research Laboratories (MERL), Cambridge, MA.  He has held postdoctoral
positions in the Wireless Foundations, University of California,
Berkeley, and in the Information Processing Laboratory, University of
Toronto.  He has worked at Disney Research, Cambridge, MA, Arraycomm, San Jose, CA, the C. S. Draper
Laboratory, Cambridge, MA, and Ktaadn, Newton, MA.  His research
interests include communication and information theory,
error-correction coding, statistical signal processing and
optimization, security, and application of these disciplines to
computer architecture.

Dr.~Draper has received the NSERC Discovery Award, the NSF CAREER
Award, the 2010 MERL President's Award, departmental teaching awards
from the University of Toronto and the University of Wisconsin, the MIT Carlton
E. Tucker Teaching Award, an Intel Graduate Fellowship, Stanford's
Frederick E. Terman Engineering Scholastic Award, and a U.S. State
Department Fulbright Fellowship.
\end{IEEEbiographynophoto}

\end{document}